\documentclass{article}
\usepackage[lang = british]{ems-aihpd} 

\newcommand{\ud}{\mathrm{d}}
\newcommand{\Tr}{\,\mathrm{Tr}}
\def\smalllozenge{\mbox{\scriptsize$\lozenge$}}
\def\smallblacklozenge{\mbox{\scriptsize$\blacklozenge$}}
\makeatletter
 \let\old@rule\@rule
 \def\@rule[#1]#2#3{\textcolor{rulecolor}{\old@rule[#1]{#2}{#3}}}
 \definecolor{rulecolor}{named}{black}
\makeatother

\newtheoremstyle{mydefinition}{3pt}{3pt}{\sf}{}{\bf}{.}{.5em}{}
\theoremstyle{mydefinition}
\newtheorem{dfnt}{Definition}[section]

\theoremstyle{plain}

\newtheorem{thrm}[dfnt]{Theorem}
\newtheorem{prps}[dfnt]{Proposition}
\newtheorem{corl}[dfnt]{Corollary}

\theoremstyle{remark}
\newtheorem{rmk}[dfnt]{Remark}
\newtheorem{exm}[dfnt]{Example}

\makeatletter
\renewcommand\ps@titlepage{%
        \let\@oddhead\@empty
        \let\@evenfoot\@empty
        \def\@oddfoot{\hfill\thepage}
        \let\@evenhead\@oddhead
}
\makeatother

\begin{document}
\sloppy
\vspace*{-1.9cm}

\title{Nested Catalan tables and a recurrence relation in noncommutative 
quantum field theory}
\titlemark{Nested Catalan tables and a recurrence relation in noncommutative 
QFT}

\emsauthor{1}{Jins de Jong}{J.~de Jong}
\emsauthor{2}{Alexander Hock}{A.~Hock}
\emsauthor{3}{Raimar Wulkenhaar}{R. Wulkenhaar}


\emsaffil{1}{TNO, Postbus 1416, 9701 BK Groningen, The Netherlands
  \email{jins.dejong@tno.nl}}
\emsaffil{2}{Mathematical Institute -- Andrew Wiles Building,
  University of Oxford, Woodstock Road, OX2 6GG, Oxford, United Kingdom
  \email{alexander.hock@maths.ox.ac.uk}}
\emsaffil{3}{Mathematisches Institut der Westf\"alischen Wilhelms-Universit\"at,
Einsteinstr. 62, 48149 M\"unster, Germany \email{raimar@math.uni-muenster.de}}



\begin{abstract}
  Correlation functions in a dynamic quartic matrix model are 
obtained from the two-point function through a recurrence relation. 
This paper gives the explicit solution of the recurrence by mapping it
bijectively to a two-fold nested combinatorial structure
each counted by Catalan numbers. These `nested Catalan tables'
have a description as diagrams of non-crossing chords and threads.
\end{abstract}

\maketitle

\section{Introduction}

The quartic matrix model is defined by the following measure 
on the space of self-adjoint 
$\mathcal{N}\times \mathcal{N}$-matrices:
\begin{align}
\label{measure}
\ud\mu(\Phi)= \frac{1}{\mathcal{Z}} 
\exp\Big(-\mathcal{N} \Tr\Big(E \Phi^2+\frac{\lambda}{4} \Phi^4
\Big)\Big)\;\ud\Phi\;,
\end{align}
where $E=\mathrm{diag}(E_0,...,E_{\mathcal{N}-1})$ is a positive
$\mathcal{N}\times \mathcal{N}$-matrix,
$\lambda$ a scalar and $\ud\Phi$ the standard Lebesgue measure.
The measure  (\ref{measure}) gives rise to
moments 
$\langle a_1b_1;...;a_Nb_N\rangle:= \int \ud\mu(\Phi)\; 
\Phi_{a_1b_1}
\Phi_{a_2b_2}\cdots 
\Phi_{a_Nb_N}$ which decompose as usual into cumulants
$\langle a_1b_1;...;a_Nb_N\rangle_c$.

This matrix model arises from a programme to understand Euclidean
quantum fields on noncommutative spaces \cite{SurveyNCG}.  The
large-$\mathcal{N}$ limit of properly rescaled cumulants
$\langle a_1b_1;...;a_Nb_N\rangle_c$, in a suitable topology, leads to
the same challenges as in familiar quantum field theories concerning
renormalisation and existence for $\lambda \neq 0$. It turned out that
for the matrix model the challenges are easier to master. Consider
cumulants with pairwise different $a_i$. Then
$\langle a_1b_1;...;a_Nb_N\rangle_c$ is only non-vanishing if $N$ is
even and $b_i=a_{\pi(i)}$ for some permutation
$\pi \in \mathcal{S}_N$. If $c(\pi)$ is the number of cycles in $\pi$,
we expand
\begin{align}
  \mathcal{N}^N \langle a_1b_{\pi(1)};...;a_Nb_{\pi(N)}\rangle_c
=:\sum_{g=0}^\infty \mathcal{N}^{2-2g-c(\pi)} G^{(g,\pi)}_{a_1....a_N}\;.
\end{align}
This paper is part of the programme to construct functions
$Z(\mathcal{N},\lambda)$, $\mu^2(\mathcal{N},\lambda)$ such that when starting
from (\ref{measure}) with
\[
E_k\mapsto Z(\mathcal{N},\lambda) \Big(E_k
+\frac{1}{2} \mu^2(\mathcal{N},\lambda)\Big)\;,\qquad \lambda\mapsto
\big(Z(\mathcal{N},\lambda)\big)^2\lambda\;,
\]
every limit $ \lim_{\mathcal{N}\to\infty} G^{(g,\pi)}_{a_1....a_N}$
exists in a neighbourhood of $\lambda=0$. 

The first step consists in understanding the case where $\pi$ has a
single cycle $c(\pi)=1$ and in leading order $g=0$ of the
$1/\mathcal{N}$-expansion. We relabel the indices to achieve
$b_{i}=a_{i+1}$ (with $b_0\equiv b_N$) and write
$G^{(0,\pi(i)=i+1)}_{a_1...a_N}=G^{(0)}_{b_0...b_{N-1}}$. For these
functions the following recurrence relation was proved in
\cite{Grosse:2012uv}:
\begin{equation}
G^{(0)}_{b_{0}\ldots b_{N-1}}=-\lambda\sum_{l=1}^{\frac{N-2}{2}}
\frac{G^{(0)}_{b_{0}\ldots b_{2l-1}}\cdot G^{(0)}_{b_{2l}\ldots b_{N-1}} 
- G^{(0)}_{b_{1}\ldots b_{2l}}\cdot G^{(0)}_{b_{0}b_{2l+1}\ldots b_{N-1}}}{
(E_{b_{0}}-E_{b_{2l}})(E_{b_{1}}-E_{b_{N-1}})}\;.\label{e:rr}
\end{equation}
Equation (\ref{e:rr}) is the counterpart of \emph{Tutte
  equations} arising in the enumeration of maps on surfaces
\cite{Tutte} or of \emph{loop equations} in matrix models
\cite{Eynard:2016yaa}.
The recurrence relation \eqref{e:rr} is specific to the measure
(\ref{measure}); Dyson-Schwinger techniques and $U(\mathcal{N})$
invariance of the partition function are used to prove it.
The planar 2-point function $G^{(0)}_{b_0b_1}$ satisfies 
a closed non-linear equation \cite{Grosse:2009pa} which was 
solved in \cite{Panzer:2018tvy} for a limiting case of
linearly spaced $E_k=(c_0+kc_1)$.

In this paper we establish a bijection between the solution of
(\ref{e:rr}) and a combinatorial problem for two nested structures
each counted by Catalan numbers. We thus propose to name them `nested
Catalan tables'.  As by-product we observed that the same relation
\eqref{e:rr} appears in the planar sector of the 2-matrix model for
mixed correlation functions \cite{Eynard:2005iq}.  The distinction
between even $b_{2i}$ and odd $b_{2i+1}$ matrix indices in
\eqref{e:rr} corresponds to the different matrices of the 2-matrix
model. This observation together with a striking r\^{o}le of an
involution in \cite{Panzer:2018tvy} supported the conjecture that also
the quartic matrix model (\ref{measure}) relates to topological
recursion \cite{Eynard:2007kz, Eynard:2016yaa}. This vision led two of
us (AH+RW) together with H.~Grosse in \cite{Grosse:2019jnv} to an
exact solution $G^{(0)}_{b_0b_1}$ of the non-linear equation
\cite{Grosse:2009pa} for arbitrary $E_k$ and $\lambda$ in or near
$\mathbb{R}_+$. Together with results of this paper we thus have a
complete understanding of the cumulants
$G^{(0,\pi(i)=i+1)}_{a_1...a_N}=G^{(0)}_{b_0...b_{N-1}}$ in leading
$\frac{1}{\mathcal{N}}$-order. In the meantime a precise relation
between (\ref{measure}) and blobbed topological recursion
\cite{Borot:2015hna} was established in \cite{Branahl:2020yru,
  Branahl:2020uxs, Hock:2021tbl}. This means that the quartic matrix
model generates the combinatorics of a family of intersection numbers
of characteristic classes on the moduli space
$\overline{\mathcal{M}}_{g,n}$ of stable complex curves.

\medskip

Let us return to the recurrence relation \eqref{e:rr} and explain
the combinatorial problem. For this purpose it is
safe to consider the $\{E_{b_{j}}\}$ as pairwise different formal
variables and to set $\lambda=-1$.
The complete expression for the ($N=2k+2$)-point function 
$G^{(0)}_{b_{0}b_{1}\ldots b_{2k+1}}$ 
according to (\ref{e:rr}) yields $2^kc_k$ terms of the form
\begin{equation}
\frac{\pm G^{(0)}_{b_{p}b_{q}}\cdots G^{(0)}_{b_{r}b_{s}}}{
(E_{b_{t}}-E_{b_{u}})\cdots(E_{b_{v}}-E_{b_{w}})}
\label{expansion}
\end{equation}
with $p<q$, $r<s$, $t<u$ and $v<w$, where
$c_k=\frac{1}{k+1}\binom{2k}{k}$ is the $k$th Catalan number.
However, some of the terms cancel. 
In this paper we answer the so far open questions: \emph{Which terms 
survive the cancellations? Can they be explicitly characterised, without 
going into the recursion?} The answer will be encoded in \emph{nested Catalan
tables}.

The paper is organised as follows.  In Sec.~\ref{sec:sym} the 
symmetries of $G^{(0)}_{b_0...b_{N-1}}$ are discussed. Afterwards, 
in Secs.~\ref{sec:CT} and \ref{sec:Ctab} we introduce 
Catalan tuples, nested Catalan tables, certain trees and operations on them.
The Catalan numbers
$c_{k}=\frac{1}{k+1}\binom{2k}{k}$ will count various parts of our
results and will be related to the number
$d_{k}=\frac{1}{k+1}\binom{3k+1}{k}$ of 
nested Catalan tables of length $k+1$, see Proposition \ref{thrm:NrT}.
Sec.~\ref{sec:bijection} is the main part of this paper. We prove in 
Theorem \ref{thrm:RC} that nested Catalan tables precisely encode the 
surviving terms in the expansion of $G^{(0)}_{b_0...b_{N-1}}$ with 
specified designated node. 

Both the nested Catalan tables and the $G^{(0)}_{b_0...b_{N-1}}$ can be depicted conveniently 
as chord diagrams with threads, which will be
introduced in Appendix~\ref{sec:CD}. 
Through these diagrams it will become clear that the recursion 
relation (\ref{e:rr}) is related to well-known combinatorial
problems~\cite{noy1,noy2}.

\section{Symmetries\label{sec:sym}}

The two-point function is symmetric, $G^{(0)}_{b_{p}b_{q}}=G^{(0)}_{b_{q}b_{p}}$. 
Because there is an even number of antisymmetric factors in 
the denominator of each term, it follows immediately that
\begin{equation}
G^{(0)}_{b_{0}b_{1}\ldots b_{N-1}}=G^{(0)}_{b_{N-1}\ldots b_{1}b_{0}}\;.\label{e:rr_s}
\end{equation}
Our aim is to prove cyclic invariance $G^{(0)}_{b_{0}b_{1}\ldots
  b_{N-1}}=G^{(0)}_{b_{1}\ldots b_{N-1}b_0}$.
We proceed by induction. Assuming that all $n$-point functions 
with $n\leq N-2$ are cyclically invariant, it is not difficult to check that
\begin{align}
&\hspace{-8mm}G^{(0)}_{b_{0}b_{1}\ldots b_{N-1}}
=\sum_{l=1}^{\frac{N-2}{2}}\frac{G^{(0)}_{b_{0}\ldots b_{2l-1}}\cdot 
G^{(0)}_{b_{2l}\ldots b_{N-1}} 
- G^{(0)}_{b_{1}\ldots b_{2l}}\cdot G^{(0)}_{b_{0}b_{2l+1}\ldots  b_{N-1}}}{
(E_{b_{0}}-E_{b_{2l}})(E_{b_{1}}-E_{b_{N-1}})}\nonumber
\\
&=-\sum_{l=1}^{\frac{N-2}{2}}\frac{G^{(0)}_{b_{0}b_{N-1}\ldots
    b_{2l+1}}\cdot 
G^{(0)}_{b_{2l}\ldots b_{1}}-G^{(0)}_{b_{N-1}\ldots b_{2l}}
\cdot G^{(0)}_{b_{2l-1}\ldots b_{1}b_{0}}}{
(E_{b_{0}}-E_{b_{2l}})(E_{b_{1}}-E_{b_{N-1}})}\nonumber
\\
&=\sum_{k=1}^{\frac{N-2}{2}}\frac{G^{(0)}_{b_{0}b_{N-1}\ldots
    b_{N-2k+1}}\cdot G^{(0)}_{b_{N-2k}\ldots b_{1}} 
- G^{(0)}_{b_{N-1}\ldots b_{N-2k}}\cdot 
G^{(0)}_{b_0 b_{N-2k-1}\ldots b_{1}}}{
(E_{b_{0}}-E_{b_{N-2k}})(E_{b_{N-1}}-E_{b_{1}})}\nonumber
\\
&=G^{(0)}_{b_{0}b_{N-1}\ldots b_{1}}
=G^{(0)}_{b_{1}\ldots b_{N-1}b_0}\;.\label{e:rr_rot}
\end{align}
The transformation $2l=N-2k$ and the symmetry (\ref{e:rr_s}) are 
applied here to rewrite the sum. This shows cyclic invariance.

Although the $N$-point functions are invariant under a cyclic
permutation of its indices, the preferred expansion into surving terms 
(\ref{expansion}) will depend on the choice of a designated node
$b_{0}$, the root. Our preferred expansion will have a clear combinatorial
significance, but it cannot be unique because of 
\begin{align}
\hspace*{-1em}
\frac{1}{E_{b_p}{-}E_{b_q}}\cdot\frac{1}{E_{b_q}{-}E_{b_r}}
+\frac{1}{E_{b_r}{-}E_{b_p}}\cdot\frac{1}{E_{b_p}{-}E_{b_q}}
+\frac{1}{E_{b_q}{-}E_{b_r}}\cdot\frac{1}{E_{b_r}{-}E_{b_p}}=0\;.
\label{rfp}
\end{align}
These identities must be employed several times to establish cyclic 
invariance of our preferred expansion.

\section{Catalan tuples\label{sec:CT}}

\begin{dfnt}[Catalan tuple]\label{dfnt:cattup}
  A Catalan tuple $\tilde{e}=(e_{0},\ldots,e_{k})$ of length
  $k\in\mathbb{N}_{0}$ is a tuple of integers $e_{j}\geq0$ for
  $j=0,\ldots,k$, such that
\begin{equation}
\sum_{j=0}^{k}e_{j}=k\qquad \text{ and }\qquad 
\sum_{j=0}^{l}e_{j}>l\quad\text{for }l=0,\ldots,k-1\;.
\label{ineq-Catalan}
\end{equation}
The set of Catalan tuples of length $|\tilde{e}|:=k$ is denoted by
$\mathcal{C}_{k}$. 
\end{dfnt}
\noindent
For $\tilde{e}=(e_{0},\ldots,e_{k})$
it follows immediately that, for all $k\geq 0$, $e_{k}=0$ and that,
for all $k>0$, $e_{0}>0$.

\begin{exm}
We have $\mathcal{C}_{0}=\{(0)\}$, $\mathcal{C}_{1}=\{(1,0)\}$ and 
$\mathcal{C}_{2}=\{(2,0,0),(1,1,0)\}$. 
  All Catalan tuples of length $3$ are given in the first column of
  Table~\ref{f:CTk3}. 
\end{exm}

\begin{rmk}
  Catalan tuples can be used to establish bijections with
  several structures counted
  by Catalan numbers. In Definitions~\ref{dfnt:rpt} and \ref{dfnt:ot}
  we provide two bijections to
  planted plane trees. Here we give the bijection to Dyck paths
  on a $k\times k$ lattice which do not go below the diagonal.
  Given a Catalan tuple $\tilde{e}=(e_0,\ldots,e_k)$ with $k\geq 1$. Start
  at the bottom left corner, 
go $e_0$ steps north followed by one step east, then 
go $e_1$ steps north followed by one step east, \dots,
finally go $e_{k-1}$ steps north followed by one step east, and stop
at the top right corner. The first condition in (\ref{ineq-Catalan})
prevents the path from going below the diagonal, the last condition
guarantees that the path ends at the top right corner.
The last row of Table~\ref{f:CTk3} gives
the Dyck paths for the Catalan tuples of length $3$.
\end{rmk}

We now define two particular compositions of Catalan tuples.
Appendix~\ref{app:ex} provides a few examples.
\begin{dfnt}[$\circ$-composition]\label{dfnt:circ}
The composition $\circ:\mathcal{C}_{k}
\times\mathcal{C}_{l}\rightarrow \mathcal{C}_{k+l+1}$ is given by
\begin{align*}
(e_{0},\ldots,e_{k})\circ(f_{0},\ldots,f_{l})
&:=(e_{0}+1,e_{1},\ldots,e_{k-1},e_{k},f_{0},f_{1},\ldots,f_{l})\;.
\end{align*}
\end{dfnt}
\noindent
No information is lost in this composition, i.e.\ it is possible to 
uniquely retrieve both terms. In particular, $\circ$ cannot be 
associative or commutative. Consider for a Catalan tuple 
$\tilde{e}=(e_{0},\ldots,e_{k})$ partial sums 
$p_{l}:\mathcal{C}_{k}\rightarrow \{0,\ldots,k\}$ and maps 
$\sigma_{a}:\mathcal{C}_{k}\rightarrow \{0,\ldots,k\}$ defined by
\begin{align}
p_{l}(\tilde{e}) &:=-l+\sum_{j=0}^{l}e_{j}\;,\qquad\text{for } 
l=0,\ldots,k-1\;,
\\
\sigma_{a}(\tilde{e})&:=\min\{l\,|\,p_{l}(\tilde{e})=a\}\;.
\nonumber
\end{align}
Then 
\begin{equation}
\tilde{e}=(e_{0},\ldots,e_{k})
= (e_{0}-1,e_{1},\ldots,e_{\sigma_{1}(\tilde{e})})
\circ(e_{\sigma_{1}(\tilde{e})+1},\ldots,e_{k})\;.
\label{circ-factor}
\end{equation}
Because $\sigma_{1}(\tilde{e})$ exists for any
$\tilde{e}\in\mathcal{C}_{k}$ with $k\geq1$, every Catalan tuple has
unique $\circ$-factors. Only these two Catalan tuples, composed by
$\circ$, yield $(e_{0},\ldots,e_{k})$. This implies that the number
$c_k$ of Catalan tuples in $\mathcal{C}_{k}$ satisfies Segner's recurrence
relation
\begin{equation*}
c_{k}=\sum_{m=0}^{k-1}c_{m}c_{k-1-m}
\end{equation*}
together with $c_0=1$, which is solved by the Catalan numbers
$c_k=\frac{1}{k+1}\binom{2k}{k}$.

In Remark \ref{Dyck-circ} we formulate the $\circ$-decomposition
in terms of Dyck paths.

The other composition of Catalan tuples is a variant of the $\circ$-product. 
\begin{dfnt}[$\bullet$-composition]\label{dfnt:bullet}
The composition $\bullet:\mathcal{C}_{k}\times\mathcal{C}_{l}
\rightarrow \mathcal{C}_{k+l+1}$ is given by
\begin{equation*}
(e_{0},\ldots,e_{k})\bullet(f_{0},\ldots,f_{l})=(e_{0}+1,f_{0},\ldots,f_{l},e_{1},\ldots,e_{k})\quad.
\end{equation*}
\end{dfnt}
As in the case of the composition $\circ$, Definition \ref{dfnt:circ},
no information is lost in the product $\bullet$. It is reverted by
\begin{equation}
\tilde{e}=(e_0,\ldots,e_k)
=(e_0-1,e_{1+\sigma_{e_{0}-1}(\tilde{e})},\ldots,e_k) 
\bullet (e_1,\ldots ,e_{\sigma_{e_{0}-1}(\tilde{e})})\;.
\label{bullet-factor}
\end{equation}
Because $\sigma_{e_{0}-1}(\tilde{e})$ exists for any 
$\tilde{e}\in\mathcal{C}_{k}$ with $k\geq 1$
(also for $e_0=1$ where $\sigma_{e_0-1}(\tilde{e}) = k$),
every Catalan tuple has a unique pair of $\bullet$-factors.

In Remark \ref{Dyck-bullet} we formulate the $\bullet$-decomposition
in terms of Dyck paths.

Out of these Catalan tuples we will construct three sorts of trees: 
\emph{pocket tree}, \emph{direct tree}, \emph{opposite tree}. They are
all planted plane trees, which means they are embedded into the plane and 
planted into a monovalent 
phantom root which connects to a unique vertex that we consider as the 
(real) root. 
We adopt the convention that the phantom root 
is not shown; its implicit presence manifests in a different 
counting of the valencies of the real root.
Pocket tree and direct tree are the same, but
their r\^ole will be different. 
Their drawing algorithms are given by the next definitions.
\begin{dfnt}[direct tree, pocket tree] \label{dfnt:rpt}
For a Catalan tuple $(e_{0},\ldots,e_{k})\in\mathcal{C}_{k}$, 
draw $k+1$ vertices on a line. Starting at the root $l=0$:
\begin{itemize}
\item unless $l=0$, connect this vertex to the last vertex ($m<l$) 
with an open half-edge;
\item if $e_{l}>0$: $e_{l}$ half-edges must be attached to vertex $l$;
\item move to the next vertex.
\end{itemize}
For direct trees, vertices will be called \emph{nodes} and edges will be 
called \emph{threads}; they are oriented from left to right.
For pocket trees, vertices are called \emph{pockets}.
\end{dfnt}

\begin{dfnt}[opposite tree]\label{dfnt:ot}
For a Catalan tuple $(e_{0},\ldots,e_{k})\in\mathcal{C}_{k}$, 
draw $k+1$ vertices on a line. Starting at the root $l=0$:
\begin{itemize}
\item if $e_{l}>0$: $e_{l}$ half-edges must be attached to vertex $l$;
\item if $e_{l}=0$:
\item[-] connect vertex $l$ to the last vertex ($m<l$) with an open half-edge;
\item[-] if vertex $l$ is now not connected to the 
last vertex ($n\leq m<l$) with an open half-edge, repeat this until it is;
\item move to the next vertex.
\end{itemize}
For opposite trees, vertices will be called \emph{nodes} and edges
will be called \emph{threads}; they are oriented from left to right.
\end{dfnt}

Examples of these trees can be seen in Figure \ref{f:PT} and
Table~\ref{f:CTk3}. It will be explained in Sec.~\ref{sec:bijection} how these
trees relate to the recurrence relation (\ref{e:rr}) and how to 
label the nodes. 
The pocket trees will often be represented
with a top-down orientation, instead of a left-right one.

\begin{figure}[!ht]
\setlength{\unitlength}{0.75mm}
\begin{picture}(100,42)
\put(-27,32){\mbox{DT:}}
\put(-27,5){\mbox{OT:}}
\multiput(-11,29)(10,0){15}{\textbullet}
\put(-5,30){\oval(10,5)[t]}
\put(0,30){\oval(20,10)[t]}
\put(5,30){\oval(30,15)[t]}
\put(25,30){\oval(10,5)[t]}
\put(35,30){\oval(10,5)[t]}
\put(40,30){\oval(20,10)[t]}
\put(45,30){\oval(30,15)[t]}
\put(30,30){\oval(80,20)[t]}
\put(75,30){\oval(10,5)[t]}
\put(85,30){\oval(10,5)[t]}
\put(90,30){\oval(20,10)[t]}
\put(90,30){\oval(40,15)[t]}
\put(55,30){\oval(130,25)[t]}
\put(60,30){\oval(140,30)[t]}

\multiput(-11,-1)(10,0){15}{\textbullet}
\put(-5,0){\oval(10,5)[t]}
\put(0,0){\oval(20,10)[t]}
\put(25,0){\oval(70,25)[t]}
\put(50,0){\oval(120,30)[t]}
\put(55,0){\oval(130,35)[t]}
\put(60,0){\oval(140,40)[t]}
\put(40,0){\oval(40,20)[t]}
\put(35,0){\oval(10,5)[t]}
\put(40,0){\oval(20,10)[t]}
\put(45,0){\oval(30,15)[t]}
\put(90,0){\oval(40,20)[t]}
\put(85,0){\oval(30,15)[t]}
\put(85,0){\oval(10,5)[t]}
\put(90,0){\oval(20,10)[t]}

\end{picture}
\caption{Direct tree (upper) and the opposite tree (lower) for the 
Catalan tuple $(6,0,0,1,3,0,0,0,2,2,0,0,0,0,0)
=(5,0,0,1,3,0,0,0,2,2,0,0,0,0)\circ(0)
=(5,0,1,3,0,0,0,2,2,0,0,0,0,0)\bullet (0)$.\label{f:PT}}
\end{figure}
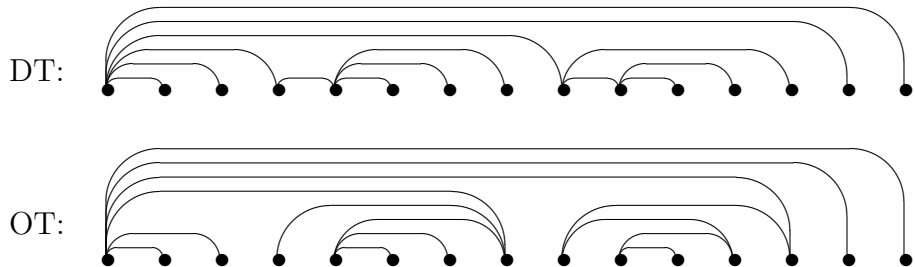

\begin{table}[!ht]
$\begin{array}{|c|c|c|c|c|}
\hline
{\text{Catalan tuple}} &
\text{pocket tree} &
\text{direct tree} &
\text{opposite tree} &
\text{Dyck path} 
   \\ \hline 
\raisebox{5mm}{(3,0,0,0)} 
&
\begin{picture}(20,13)
\put(10.4,7){\textbullet}
\put(5.4,2){\textbullet}
\put(10.4,2){\textbullet}
\put(15.4,2){\textbullet}
\put(11,8){\line(-1,-1){5}}
\put(11,8){\line(1,-1){5}}
\put(11,8){\line(0,-1){5}}
\end{picture}
&
\begin{picture}(27,14)
\put(5.4,2){\textbullet}
\put(10.4,2){\textbullet}
\put(15.4,2){\textbullet}
\put(20.4,2){\textbullet}
\put(8.5,3){\oval(5,5)[t]}
\put(11,3){\oval(10,10)[t]}
\put(13.5,3){\oval(15,15)[t]}
\end{picture}
&
\begin{picture}(27,14)
\put(5.4,2){\textbullet}
\put(10.4,2){\textbullet}
\put(15.4,2){\textbullet}
\put(20.4,2){\textbullet}
\put(8.5,3){\oval(5,5)[t]}
\put(11,3){\oval(10,10)[t]}
\put(13.5,3){\oval(15,15)[t]}
\end{picture}
&
\begin{picture}(15,17)                
\put(0,0){\line(1,0){15}}
\put(0,5){\line(1,0){15}}
\put(0,10){\line(1,0){15}}
\put(0,15){\line(1,0){15}}
\put(0,0){\line(0,1){15}}
\put(5,0){\line(0,1){15}}
\put(10,0){\line(0,1){15}}
\put(15,0){\line(0,1){15}}
\linethickness{0.7mm}
\put(0,0){\line(0,1){15}}
\put(0,15){\line(1,0){15}}
\end{picture}
\\ \hline 
\raisebox{6mm}{(2,1,0,0)} 
&
\begin{picture}(20,14)
\put(10.4,10){\textbullet}
\put(5.4,5){\textbullet}
\put(15.4,5){\textbullet}
\put(5.4,0){\textbullet}
\put(11,11){\line(-1,-1){5}}
\put(11,11){\line(1,-1){5}}
\put(6,6){\line(0,-1){5}}
\end{picture}
&
\begin{picture}(27,13)
\put(5.4,2){\textbullet}
\put(10.4,2){\textbullet}
\put(15.4,2){\textbullet}
\put(20.4,2){\textbullet}
\put(8.5,3){\oval(5,5)[t]}
\put(13.5,3){\oval(5,5)[t]}
\put(13.5,3){\oval(15,10)[t]}
\end{picture}
&
\begin{picture}(27,14)
\put(5.4,2){\textbullet}
\put(10.4,2){\textbullet}
\put(15.4,2){\textbullet}
\put(20.4,2){\textbullet}
\put(11,3){\oval(10,10)[t]}
\put(13.5,3){\oval(5,5)[t]}
\put(13.5,3){\oval(15,15)[t]}
\end{picture}
&
\begin{picture}(15,17)                
\put(0,0){\line(1,0){15}}
\put(0,5){\line(1,0){15}}
\put(0,10){\line(1,0){15}}
\put(0,15){\line(1,0){15}}
\put(0,0){\line(0,1){15}}
\put(5,0){\line(0,1){15}}
\put(10,0){\line(0,1){15}}
\put(15,0){\line(0,1){15}}
\linethickness{0.7 mm}
\put(0,0){\line(0,1){10}}
\put(0,10){\line(1,0){5}}
\put(5,10){\line(0,1){5}}
\put(5,15){\line(1,0){10}}
\end{picture}
 \\ \hline 
\raisebox{6mm}{(2,0,1,0)} 
&
\begin{picture}(20,14)
\put(10.4,10){\textbullet}
\put(5.4,5){\textbullet}
\put(15.4,5){\textbullet}
\put(15.4,0){\textbullet}
\put(11,11){\line(-1,-1){5}}
\put(11,11){\line(1,-1){5}}
\put(16,6){\line(0,-1){5}}
\end{picture}
&
\begin{picture}(27,13)
\put(5.4,2){\textbullet}
\put(10.4,2){\textbullet}
\put(15.4,2){\textbullet}
\put(20.4,2){\textbullet}
\put(8.5,3){\oval(5,5)[t]}
\put(18.5,3){\oval(5,5)[t]}
\put(11,3){\oval(10,10)[t]}
\end{picture}
&
\begin{picture}(27,13)
\put(5.4,2){\textbullet}
\put(10.4,2){\textbullet}
\put(15.4,2){\textbullet}
\put(20.4,2){\textbullet}
\put(8.5,3){\oval(5,5)[t]}
\put(18.5,3){\oval(5,5)[t]}
\put(13.5,3){\oval(15,10)[t]}
\end{picture}
&
\begin{picture}(15,17)                
\put(0,0){\line(1,0){15}}
\put(0,5){\line(1,0){15}}
\put(0,10){\line(1,0){15}}
\put(0,15){\line(1,0){15}}
\put(0,0){\line(0,1){15}}
\put(5,0){\line(0,1){15}}
\put(10,0){\line(0,1){15}}
\put(15,0){\line(0,1){15}}
\linethickness{0.7 mm}
\put(0,0){\line(0,1){10}}
\put(0,10){\line(1,0){10}}
\put(10,10){\line(0,1){5}}
\put(10,15){\line(1,0){5}}
\end{picture}
\\
\hline
\raisebox{6mm}{(1,2,0,0)} 
&
\begin{picture}(20,13)
\put(10.3,6){\textbullet}
\put(5.4,1){\textbullet}
\put(10.3,11){\textbullet}
\put(15.3,1){\textbullet}
\put(11,7){\line(-1,-1){5}}
\put(10.8,7){\line(1,-1){5}}
\put(10.8,12){\line(0,-1){5}}
\end{picture}
&
\begin{picture}(27,15)
\put(5.4,2){\textbullet}
\put(10.4,2){\textbullet}
\put(15.4,2){\textbullet}
\put(20.4,2){\textbullet}
\put(8.5,3){\oval(5,5)[t]}
\put(13.5,3){\oval(5,5)[t]}
\put(16,3){\oval(10,10)[t]}
\end{picture}
&
\begin{picture}(27,13)
\put(5.4,2){\textbullet}
\put(10.4,2){\textbullet}
\put(15.4,2){\textbullet}
\put(20.4,2){\textbullet}
\put(16,3){\oval(10,10)[t]}
\put(13.5,3){\oval(5,5)[t]}
\put(13.5,3){\oval(15,15)[t]}
\end{picture}
&
\begin{picture}(15,17)                
\put(0,0){\line(1,0){15}}
\put(0,5){\line(1,0){15}}
\put(0,10){\line(1,0){15}}
\put(0,15){\line(1,0){15}}
\put(0,0){\line(0,1){15}}
\put(5,0){\line(0,1){15}}
\put(10,0){\line(0,1){15}}
\put(15,0){\line(0,1){15}}
\linethickness{0.7 mm}
\put(0,0){\line(0,1){5}}
\put(0,5){\line(1,0){5}}
\put(5,5){\line(0,1){10}}
\put(5,15){\line(1,0){10}}
\end{picture}
\\ \hline 
\raisebox{8mm}{(1,1,1,0)} 
&
\begin{picture}(20,18)
\put(10.4,15){\textbullet}
\put(10.4,10){\textbullet}
\put(10.4,5){\textbullet}
\put(10.4,0){\textbullet}
\put(11,15.5){\line(0,-1){15}}
\end{picture}
&
\begin{picture}(27,18)
\put(5.4,5){\textbullet}
\put(10.4,5){\textbullet}
\put(15.4,5){\textbullet}
\put(20.4,5){\textbullet}
\put(8.5,6){\oval(5,5)[t]}
\put(13.5,6){\oval(5,5)[t]}
\put(18.5,6){\oval(5,5)[t]}
\end{picture}
&
\begin{picture}(27,18)
\put(5.4,5){\textbullet}
\put(10.4,5){\textbullet}
\put(15.4,5){\textbullet}
\put(20.4,5){\textbullet}
\put(13.5,6){\oval(15,15)[t]}
\put(16,6){\oval(10,10)[t]}
\put(18.5,6){\oval(5,5)[t]}
\end{picture}
&
\begin{picture}(15,17)                
\put(0,0){\line(1,0){15}}
\put(0,5){\line(1,0){15}}
\put(0,10){\line(1,0){15}}
\put(0,15){\line(1,0){15}}
\put(0,0){\line(0,1){15}}
\put(5,0){\line(0,1){15}}
\put(10,0){\line(0,1){15}}
\put(15,0){\line(0,1){15}}
\linethickness{0.7 mm}
\put(0,0){\line(0,1){5}}
\put(0,5){\line(1,0){5}}
\put(5,5){\line(0,1){5}}
\put(5,10){\line(1,0){5}}
\put(10,10){\line(0,1){5}}
\put(10,15){\line(1,0){5}}
\end{picture}
 \\ \hline 
\end{array}
$
\caption{Catalan tuples, their corresponding planted plane trees and
Dyck paths for $k=3$. The phantom roots of the planted plane trees 
are not shown. The  real root is on top for the pocket tree 
and on the left for direct and opposite trees. 
\label{f:CTk3}}
\end{table}

\section{Nested Catalan tables\label{sec:Ctab}}

A nested Catalan table is a `Catalan tuple of Catalan tuples':
\begin{dfnt}[nested Catalan table]\label{dfnt:cattab}
A \emph{nested Catalan table of length $k$} is a tuple 
$T_k=\langle\tilde{e}^{(0)},\tilde{e}^{(1)},\ldots,
\tilde{e}^{(k)}\rangle$ of Catalan tuples $\tilde{e}^{(j)}$, such 
that 
$(1+|\tilde{e}^{(0)}|,|\tilde{e}^{(1)}|,\ldots,|\tilde{e}^{(k)}|)$,
the \emph{length} tuple of $T_k$, is 
itself a Catalan tuple of length $k$. We let $\mathcal{T}_{k}$
be the set of all nested Catalan tables of length $k$. 
The constituent $\tilde{e}^{(j)}$ in a nested Catalan table is called the 
$j$-th pocket.
\end{dfnt}
\noindent 
We will show in Sec.~\ref{sec:bijection} that a nested Catalan table
contains all information about individual terms in the expansion
(\ref{expansion}) of the $N$-point function $G^{(0)}_{b_0...b_{N-1}}$.
Nested Catalan tables have a graphical presentation as diagrams of
non-crossing chords with threads which we introduce in
Appendix~\ref{sec:CD}.

Recall the composition $\circ$ from Definition \ref{dfnt:circ} and the fact
that any Catalan tuple of length $\geq 1$ has a unique pair of 
$\circ$-factors. We extend $\circ$ as follows to nested Catalan tables:
\begin{dfnt}[$\smalllozenge$-operation]\label{dfnt:triangle}
The operation $\smalllozenge: \mathcal{T}_{k}\times
\mathcal{T}_{l}\rightarrow\mathcal{T}_{k+l}$ is given by
\begin{equation*}
\langle \tilde{e}^{(0)},\ldots,\tilde{e}^{(k)}\rangle 
\smalllozenge\langle \tilde{f}^{(0)},\ldots,\tilde{f}^{(l)}\rangle 
:= \langle \tilde{e}^{(0)} \circ \tilde{f}^{(0)},\tilde{e}^{(1)},
\ldots,\tilde{e}^{(k)}, \tilde{f}^{(1)},\ldots,\tilde{f}^{(l)}\rangle\;.
\end{equation*}
\end{dfnt}
\noindent
Now suppose the nested Catalan table on the right-hand side is given. If the
$0^{\mathrm{th}}$ pocket has length $\geq 1$, then it uniquely factors
into $\tilde{e}^{(0)} \circ \tilde{f}^{(0)}$. Consider 
\begin{equation}
\hat{k}= \sigma_{1+|\tilde{f}^{(0)}|}
\big((1+|\tilde{e}^{(0)}\circ
\tilde{f}^{(0)}|,|\tilde{e}^{(1)}|,\ldots,
|\tilde{e}^{(k)}|, |\tilde{f}^{(1)}|,\ldots,|\tilde{f}^{(l)}|)\big)\;.
\label{k-lozenge}
\end{equation}
By construction, $\hat{k}=k$ so that $\smalllozenge$ can 
be uniquely reverted. Note also that nested Catalan tables 
$\langle (0),\tilde{e}_1,\dots,\tilde{e}_k\rangle$ do not have a 
$\smalllozenge$-decomposition.

The composition $\bullet$ of Catalan tuples is extended as follows to 
nested Catalan tables:
\begin{dfnt}[$\smallblacklozenge$-operation]\label{dfnt:box}
The operation $\smallblacklozenge: 
\mathcal{T}_{k}\times\mathcal{T}_{l}\rightarrow\mathcal{T}_{k+l}$ is given by
\begin{equation*}
\langle \tilde{e}^{(0)},\ldots,\tilde{e}^{(k)}\rangle 
\smallblacklozenge\langle
\tilde{f}^{(0)},\ldots,\tilde{f}^{(l)}\rangle 
:= \langle \tilde{e}^{(0)} ,\tilde{e}^{(1)}\bullet\tilde{f}^{(0)}, 
\tilde{f}^{(1)},\ldots,\tilde{f}^{(l)}, \tilde{e}^{(2)},\ldots,
\tilde{e}^{(k)}\rangle\;.
\end{equation*}
\end{dfnt}
\noindent
If the $1^{\text{st} }$ pocket has length $\geq 1$, it uniquely
factors as $\tilde{e}^{(1)}\bullet\tilde{f}^{(0)}$, and we extract%
\begin{equation}
\hat{l}:=\sigma_{|\tilde{e}^{(0)}|+|\tilde{e}^{(1)}|+1}
\big((1+|\tilde{e}^{(0)}| ,|\tilde{e}^{(1)}\bullet\tilde{f}^{(0)}|, 
|\tilde{f}^{(1)}|,\ldots,|\tilde{f}^{(l)}|, |\tilde{e}^{(2)}|,\ldots,
|\tilde{e}^{(k)}|)\big)\;.
\label{l-blacklozenge}
\end{equation}
By construction $\hat{l}=l$, and $\smallblacklozenge$ is uniquely
reverted. 

We let $\mathcal{S}_k=\{
\langle \tilde{e}_0 , (0),\tilde{e}_2,\dots,\tilde{e}_k\rangle
\in \mathcal{T}_k\}$ be the subset of length-$k$ nested Catalan tables having 
$(0)$ as their 1$^{\text{st}}$ pocket. The nested Catalan tables $S\in 
\mathcal{S}_k$ are precisely those which do not have a 
$\smallblacklozenge$-decomposition. The distinction 
between $\mathcal{S}_l$ and its complement in $\mathcal{T}_l$ is 
the key to determine the number of nested Catalan tables:
\begin{prps}\label{thrm:NrT}
  The set $\mathcal{T}_{k+1}$ of nested Catalan tables and its subset
  $\mathcal{S}_{k+1}$ with 1$^{\text{st}}$ pocket $(0)$ have cardinalities
 \begin{equation}
 d_{k}:=|\mathcal{T}_{k+1}|=\frac{1}{k+1}\binom{3k+1}{k}
\quad\text{and}\quad 
 h_{k}:=|\mathcal{S}_{k+1}|=\frac{1}{2k+1}\binom{3k}{k}\;.\label{e:res}
\end{equation}
\end{prps}
\begin{proof}
Let 
\[
\mathcal{D}(x):=\sum_{k=1}^\infty x^k \sum_{T\in \mathcal{T}_k} T
\qquad\text{and}\qquad
\mathcal{H}(x):=\sum_{k=1}^\infty x^k \sum_{S\in \mathcal{S}_k} S
\]
be the generating function of the set of all nested Catalan tables and of 
those having $(0)$ as their 1$^{\text{st}}$ pocket, respectively. Then
\begin{align}
\mathcal{D}(x)=\mathcal{D}(x) \smallblacklozenge \mathcal{D}(x)
+\mathcal{H}(x)
\label{cG}
\end{align}
because precisely the complements 
$\mathcal{T}_k\setminus \mathcal{S}_k$ have 
a unique $\smallblacklozenge$-decomposition. With the exception of 
$\langle(0),(0)\rangle \in \mathcal{S}_1=\mathcal{T}_1$, all 
$S=\langle \tilde{e}^0,(0),
\tilde{e}^2,\dots,\tilde{e}^k\rangle
\in \mathcal{S}_k$ with $k\geq 2$ have $|\tilde{e}^0|\geq 1$. 
Therefore, they have a unique 
$\smalllozenge$-decomposition, where the left factor necessarily belongs to 
$\mathcal{S}_l$  for some $l$:
\begin{align}
\mathcal{H}(x)=\mathcal{H}(x) \smalllozenge \mathcal{D}(x)
+ x \langle(0),(0)\rangle\;.
\label{cH}
\end{align}
Introducing the generating functions 
$D(x)=\sum_{k=0}^\infty x^{k+1} d_k$ and 
$H(x)=\sum_{k=0}^\infty x^{k+1} h_k$ 
of the cardinalities
$d_{k}=|\mathcal{T}_{k+1}|$ and 
$h_{k}=|\mathcal{S}_{k+1}|$, eqs.\ (\ref{cG}) and (\ref{cH}) 
project to quadratic relations
\begin{equation}
D(x)=D(x)\cdot D(x)+H(x)
\qquad\text{and}\qquad 
H(x)=H(x)\cdot D(x)+x\;.
\label{e:quadraticequation}
\end{equation}
Multiplying the first equation by $H(x)$ and
the second one by $D(x)$ gives $x\cdot D(x)=H^{2}(x)$, which
separates (\ref{e:quadraticequation}) into cubic relations
\begin{equation}
D(x)(1-D(x))^2=x\qquad\text{and} \qquad
\frac{H(x)}{\sqrt{x}} \Big(1-\Big(\frac{H(x)}{\sqrt{x}}\Big)^2\Big)
=\sqrt{x}\;.
\label{e:cubicequation}
\end{equation}
The assertion (\ref{e:res})
follows from the Lagrange inversion 
formula\footnote{\emph{Lagrange inversion formula.} 
Let $f,g\in x\mathbb{C}[[x]]$ be formal power series inverse to each other,
$g(f(x))=x$. Then their coefficients are related by 
$[x^n]g(x)=\frac{1}{n} [x^{-1}] \frac{1}{(f(x))^n}$. In particular, for
$f(x)=\frac{x}{\phi(x)}$ one has 
$[x^n]g(x)=\frac{1}{n} [x^{n-1}] (\phi(x))^n$.}.
To obtain the first equation (\ref{e:res}) 
one sets $f(x)=x(1-x)^2$ and 
$\phi(x)=\frac{1}{(1-x)^2}$ to get 
$d_k=\frac{1}{k+1}[x^k] \frac{1}{(1-x)^{2k+2}}$.
To obtain the second equation (\ref{e:res}) 
one sets $\sqrt{x}=y$, $f(y)=y(1-y^2)$ and 
$\phi(y)=\frac{1}{1-y^2}$ to get 
$h_k=\frac{1}{2k+1}[y^{2k}] \frac{1}{(1-y^2)^{2k+1}}
=\frac{1}{2k+1}[x^{k}] \frac{1}{(1-x)^{2k+1}}$.
\end{proof}
\begin{rmk}
More information about the integer sequences $d_{k}$ (A006013) 
and $h_{k}$ (A001764) can be found via~\cite{OEIS:d} 
and~\cite{OEIS:f}, respectively.
Equations (\ref{e:cubicequation}) are 
higher-order variants of the equation $C(x)(1-C(x))=x$ for the 
generating function $C(x)=\sum_{n=0}^\infty c_nx^{n+1}$ of Catalan numbers.
\end{rmk}
\begin{corl}
The number $d_k$ of nested Catalan tables satisfies
\begin{equation}
d_{k}=\sum_{(e_{0},\ldots,e_{k+1})\in\mathcal{C}_{k+1}} 
c_{e_{0}-1}c_{e_{1}}\cdots c_{e_{k}}c_{e_{k+1}}\;. \label{dsumc}
\end{equation}
\end{corl}
\begin{proof}
There are $c_{|\tilde{e}_{0}|}\cdots c_{|\tilde{e}_{k+1}|}$ 
nested Catalan tables $\langle \tilde{e}_0,\dots,\tilde{e}_{k+1}\rangle$ 
of the same length tuple 
$(|\tilde{e}_{0}|+1,|\tilde{e}_{1}|,
\ldots,|\tilde{e}_{k+1}|)\in\mathcal{C}_{k+1}$. Set 
$e_0=|\tilde{e}_{0}|+1$ and 
$e_j=|\tilde{e}_{j}|$ for $j=1,\dots,k+1$.
\end{proof}

\section{The bijection between nested
  Catalan tables and contributions
  to $G^{(0)}_{b_0...b_{N-1}}$}
\label{sec:bijection}

This section is the main part of this paper. We will omit in the sequel the
superscript $G^{(0)}_{b_1b_2}=G_{b_1b_2}$. We remark that
the graphical presentation given in Appendix~\ref{sec:CD} was very helpful
to identify this bijection.
\begin{dfnt}\label{dfnt:CTG}
  To a nested Catalan table
  $T_{k+1}=\langle\tilde{e}^{(0)},\tilde{e}^{(1)},\ldots,
  \tilde{e}^{(k+1)}\rangle\in\mathcal{T}_{k+1}$ with $N/2=k+1$ we
  associate a monomial $[T]_{b_{0},\ldots,b_{N-1}}$ in
  $G_{b_{l}b_{m}}$ and $\frac{1}{E_{b_{l'}}-E_{b_{m'}}}$ as follows:
\begin{enumerate}
\item Build the pocket tree for the length tuple $(
  1+|\tilde{e}^{(0)}|,|\tilde{e}^{(1)}|,\ldots,
  |\tilde{e}^{(k+1)}|) \in \mathcal{C}_{k+1}$. It has $k+1$ edges
  and every edge has two sides. Starting from the root and turning
  counterclockwise, label the edge sides in consecutive
  order\footnote{This is the same order as in
    \cite[Fig.~5.14]{Stanley:1999??}.\label{fn-order}}  from
  $b_{0}$ to $b_{N-1}$. An edge labelled $b_{l}b_{m}$ encodes a factor
  $G_{b_{l}b_{m}}$ in $G^{(0)}_{b_0...b_{N-1}}$.

\item Label the $k+2$ vertices of the pocket tree by
  $P_{0},\ldots,P_{k+1}$ in consecutive order$^{\ref{fn-order}}$ when turning
  counterclockwise around the tree. Let $v(P_{m})$ be the valency of
  vertex $P_{m}$ (number of edges attached to $P_{m}$) and $L_{m}$ be the
  distance between $P_{m}$ and the root $P_{0}$ (number of edges in
  shortest path between $P_{m}$ and $P_{0}$).

\item For every vertex $P_{m}$ that is not a leaf, read off the
  $2v(P_{m})$ side labels of edges connected to $P_{m}$. Draw two
  rows of $v(P_{m})$ nodes each. Label the nodes of the
  first row by the even edge side labels in natural order,
  i.e.\ starting at the edge closest to the root and proceed in the
  counterclockwise direction. Label the nodes of the other row by
  the odd edge side labels using the same edge order. Take the $m$-th
  Catalan tuple $\tilde{e}^{(m)}$ of the nested Catalan table. 
If $L_m$ is even, draw the
direct (resp.~opposite)
tree encoded by $\tilde{e}^{(m)}$ between the row of 
even (resp.~odd) nodes.
If $L_m$ is odd, draw the
opposite (resp.~direct)
tree encoded by $\tilde{e}^{(m)}$ between the row of 
even (resp.~odd) nodes.
Encode a thread from $b_{l}$ to $b_{m}$ in the
 direct or opposite tree by a factor $\frac{1}{E_{b_{l}}-E_{b_{m}}}$.
\end{enumerate}
\end{dfnt}

\begin{rmk}\label{rmk:sidelabels}
In proofs below we sometimes have to insist that one side label of a
pocket edge is a particular $b_k$, whereas the label of the other side
does not matter. Is such a situation we will label the other side by 
$b_{\overline{k}}$. Note that if $b_k$ is 
even (resp.~odd), then  $b_{\overline{k}}$ is 
odd (resp.~even).
\end{rmk}

\begin{rmk}
  For the purpose of this article it is sufficient to mention that an
  explicit construction for the level function
  $L_{m}:\mathcal{C}_{k+1}\rightarrow\{0,\ldots, k\}$ exists.
\end{rmk}

\begin{exm}\label{ex:G12a}
Let 
$T=\langle(2,0,0),(1,1,0),(0),(0),(0),(1,0),(0)\rangle\in \mathcal{T}_6$.
Its length tuple is $(3,2,0,0,0,1,0)\in
\mathcal{C}_6$, which defines  the pocket tree: 
\[
\begin{picture}(75,33)
\put(25,27.5){\circle*{3}} 
\put(25,30){$P_{0}$}
\put(12.5,15){\circle*{3}}
\put(25,15){\circle*{3}}
\put(37.5,15){\circle*{3}}
\put(5,15){$P_{1}$}
\put(26,12){$P_{4}$}
\put(40,15){$P_{5}$}

\put(0,2.5){\circle*{3}}
\put(25,2.5){\circle*{3}}
\put(50,2.5){\circle*{3}}
\put(-7.5,2.5){$P_{2}$}
\put(27.5,2.5){$P_{3}$}
\put(52.5,2.5){$P_{6}$}

\put(25,27.5){\line(-1,-1){25}}
\put(14,21){\mbox{\small$b_{0}$}}
\put(2,9){\mbox{\small$b_{1}$}}
\put(6,5.5){\mbox{\small$b_{2}$}}
\put(17.5,17){\mbox{\small$b_{5}$}}
\put(25,27.5){\line(1,-1){25}}
\put(30.5,17){\mbox{\small$b_{8}$}}
\put(33,20.5){\mbox{\small$b_{11}$}}
\put(42.5,5){\mbox{\small$b_{9}$}}
\put(45,8){\mbox{\small$b_{10}$}}
\put(25,27.5){\line(0,-1){12.5}}
\put(21.5,20){\mbox{\small$b_{6}$}}
\put(25.5,20){\mbox{\small$b_{7}$}}
\put(12.5,15){\line(1,-1){12.5}}
\put(17,5.5){\mbox{\small$b_{3}$}}
\put(20,8.5){\mbox{\small$b_{4}$}}
\end{picture}
\]
The edge side labels encode 
\begin{equation*}
G_{b_{0}b_{5}}G_{b_{1}b_{2}}G_{b_{3}b_{4}}G_{b_{6}b_{7}} G_{b_{8}b_{11}}G_{b_{9}b_{10}}\;.
\end{equation*}
For vertex $P_{0}$, at even distance, we draw direct and opposite tree
encoded in $\tilde{e}^{(0)}=(2,0,0)$:
\begin{equation*}
\begin{picture}(80,8)
\put(0.3,2){\textbullet}
\put(10.3,2){\textbullet}
\put(20.3,2){\textbullet}
\put(-1,0){\mbox{\scriptsize$b_0$}}
\put(9,0){\mbox{\scriptsize$b_6$}}
\put(19,0){\mbox{\scriptsize$b_8$}}
\put(6,3){\oval(10,5)[t]}
\put(11,3){\oval(20,10)[t]}
\put(40.3,2){\textbullet}
\put(50.3,2){\textbullet}
\put(60.3,2){\textbullet}
\put(39,0){\mbox{\scriptsize$b_5$}}
\put(49,0){\mbox{\scriptsize$b_7$}}
\put(59,0){\mbox{\scriptsize$b_{11}$}}
\put(46,3){\oval(10,5)[t]}
\put(51,3){\oval(20,10)[t]}
\end{picture}
\end{equation*}
For vertex $P_{1}$, at odd distance, we draw opposite and direct tree
encoded in $\tilde{e}^{(1)}=(1,1,0)$:
\begin{equation*}
\begin{picture}(80,8)
\put(0.3,2){\textbullet}
\put(10.3,2){\textbullet}
\put(20.3,2){\textbullet}
\put(-1,0){\mbox{\scriptsize$b_{0}$}}
\put(9,0){\mbox{\scriptsize$b_{2}$}}
\put(19,0){\mbox{\scriptsize$b_{4}$}}
\put(16,3){\oval(10,5)[t]}
\put(11,3){\oval(20,10)[t]}
\put(40.3,2){\textbullet}
\put(50.3,2){\textbullet}
\put(60.3,2){\textbullet}
\put(39,0){\mbox{\scriptsize$b_{5}$}}
\put(49,0){\mbox{\scriptsize$b_{1}$}}
\put(59,0){\mbox{\scriptsize$b_{3}$}}
\put(56,3){\oval(10,5)[t]}
\put(46,3){\oval(10,5)[t]}
\end{picture}
\end{equation*}
For vertex $P_{5}$, at odd distance, we draw opposite and direct tree
encoded in $\tilde{e}^{(5)}=(1,0)$:
\begin{equation*}
\begin{picture}(80,5)
\put(0.3,2){\textbullet}
\put(10.3,2){\textbullet}
\put(-1,0){\mbox{\scriptsize$b_{8}$}}
\put(9,0){\mbox{\scriptsize$b_{10}$}}
\put(6,3){\oval(10,5)[t]}
\put(40.3,2){\textbullet}
\put(50.3,2){\textbullet}
\put(39,0){\mbox{\scriptsize$b_{11}$}}
\put(49,0){\mbox{\scriptsize$b_{9}$}}
\put(46,3){\oval(10,5)[t]}
\end{picture}
\end{equation*}
They give rise to a factor 
\begin{align*}
&\hspace{-8mm}\frac{1}{(E_{b_{0}}-E_{b_{6}})(E_{b_{0}}-E_{b_{8}})(E_{b_{0}}-E_{b_{4}})(E_{b_{2}}-E_{b_{4}})(E_{b_{8}}-E_{b_{10}})} \\ 
&\times\frac{1}{(E_{b_{5}}-E_{b_{7}})(E_{b_{5}}-E_{b_{11}})(E_{b_{5}}-E_{b_{1}})(E_{b_{1}}-E_{b_{3}})(E_{b_{11}}-E_{b_{9}})}\;.
\end{align*}
Later in Fig.~\ref{f:G12a} we give a diagrammatic 
representation of this nested Catalan table. 
\end{exm}

The following theorem shows that the nested Catalan tables correspond 
bijectively to the terms in the expansion of the recurrence relation 
(\ref{e:rr}).

\begin{thrm}\label{thrm:RC}
The recurrence (\ref{e:rr}) of $N$-point functions in the quartic 
matrix model (\ref{measure}) has the explicit solution 
\begin{equation*}
G^{(0)}_{b_0\ldots b_{N-1}}=\sum_{T\in \mathcal{T}_{k+1}} [T]_{b_0\ldots b_{N-1}}\;,
\end{equation*}
where the sum is over all nested Catalan tables of length $N/2=k+1$ 
and the monomials $[T]_{b_0\ldots b_{N-1}}$ are described in 
Definition \ref{dfnt:CTG}. 
\end{thrm}
\begin{proof}
  We proceed by induction in $N$. For $N=2$ the only term in the
  2-point function corresponds to the nested Catalan table $\langle
  (0),(0)\rangle\in\mathcal{T}_{1}$. Its associated length tuple $(1,0)$ 
encodes the pocket tree
\begin{equation*}
\begin{picture}(20,13)
\put(10,12){\circle*{2.5}}
\put(10,0){\circle*{2.5}}
\put(10,12){\line(0,-1){12}}
\put(6.5,6){\mbox{\scriptsize$b_0$}}
\put(11,6){\mbox{\scriptsize$b_1$}}
\end{picture}
\end{equation*}
whose single edge corresponds to a factor $G_{b_0b_1}$. 
The Catalan tuples of both pockets have length $0$, so that there is no 
denominator.

\smallskip

For any contribution to $G^{(0)}_{b_0...b_{N-1}}$ with $N\geq 4$, encoded by  
a length-$N/2$ nested Catalan table $T_{N/2}$, 
it must be shown that $T_{N/2}$ splits in one or two ways into
smaller nested Catalan tables whose corresponding monomials 
produce $T_{N/2}$ via  (\ref{e:rr}). There are three cases to
consider.

\bigskip

\noindent {\bf Case I:} Let $T_{k+1}= \langle
(0),\tilde{e}^{(1)},\ldots,\tilde{e}^{(k+1)}\rangle \in
\mathcal{T}_{k+1}$ with $N/2=k+1$. 
\\
It follows from Definition \ref{dfnt:box} that there are uniquely defined
nested Catalan tables $T_{l}= \langle
\tilde{f},\tilde{e}^{(2)},\ldots,\tilde{e}^{(l+1)}\rangle\in\mathcal{T}_l$ and $T_{k-l+1}=
\langle (0),\tilde{e},\tilde{e}^{(l+2)},\ldots,\tilde{e}^{(k+1)}\rangle\in\mathcal{T}_{k-l+1}$
with $\tilde{e}^{(1)}=\tilde{e}\bullet \tilde{f}$ and consequently
$T_{k-l+1}\smallblacklozenge T_{l}=T_{k+1}$. The length $l=\hat{l}$ is 
obtained via (\ref{l-blacklozenge}). Recall that $T_{k+1}$
cannot be obtained by the $\smalllozenge$-composition because 
the zeroth pocket has length $|(0)|=0$. By induction, $T_{l}$
encodes a unique contribution $[T_{l}]_{b_{1}\ldots b_{2l}}$ to
$G^{(0)}_{b_{1}\ldots b_{2l}}$, and $T_{k-l+1}$ encodes a unique
contribution $[T_{k-l+1}]_{b_{0}b_{2l+1}\ldots b_{N-1}}$ to
$G^{(0)}_{b_{0}b_{2l+1}\ldots b_{N-1}}$. We have to show that
\begin{equation*}
-\frac{[T_{l}]_{b_{1}\ldots b_{2l}}[T_{k-l+1}]_{b_{0}b_{2l+1}\ldots
    b_{N-1}}}{(E_{b_{0}}-E_{b_{2l}})(E_{b_{1}}-E_{b_{N-1}})}
\end{equation*}
agrees with $[T_{k+1}]_{b_{0}\ldots b_{N-1}}$ encoded by
$T_{k+1}$. A detail of the pocket tree of $T_{k+1}$ 
sketching $P_0,P_1$ and their attached edges is
\begin{align}
\parbox{80mm}{\setlength{\unitlength}{1.25mm}\begin{picture}(60,14)
\put(23,15){\circle*{2}}
\put(25,15){\mbox{\scriptsize$P_{0}$}}
\put(23,15){\line(-4,-1){16}}
\put(7,11){\circle*{2}}
\put(3,12){\mbox{\scriptsize$P_{1}$}}
\put(7,11){\line(-3,-2){10}}
\put(7,11){\line(-1,-3){3.5}}
\put(7,11){\line(1,-2){6.5}}
\put(7,11){\line(2,-1){12}}
\put(12,13.5){\mbox{\scriptsize$b_{0}$}}
\put(15,11.5){\mbox{\scriptsize$b_{N{-}1}$}}
\put(-2,7){\mbox{\scriptsize$b_{1}$}}
\put(-2,3){\mbox{\scriptsize$b_{\overline{1}}$}}
\put(12,8.5){\mbox{\scriptsize$b_{N{-}2}$}}
\put(12.8,4.3){\mbox{\scriptsize$b_{\overline{N{-}2}}$}}
\put(4,-1){\mbox{\scriptsize$b_{2l}$}}
\put(6,1.5){\mbox{\scriptsize$b_{2l{+}1}$}}
\put(13,-0.5){\mbox{\scriptsize$b_{\overline{2l{+}1}}$}}
\put(1,4){\mbox{\scriptsize$b_{\overline{2l}}$}}
\thicklines 
\qbezier[3](2,7)(2.8,6)(4.6,5.5)
\qbezier[3](10.5,5.2)(11.8,5.5)(12.8,7.3)
\end{picture}}
\end{align}
Only the gluing of the
  direct and opposite tree encoded by $\tilde{e}=(e_{0},\ldots,e_{p})$
  with the direct and opposite tree encoded by
  $\tilde{f}=(f_{0},\ldots,f_{q})$ via a thread from $b_{0}$ to
  $b_{2l}$ and a thread from $b_{N-1}$ to $b_{1}$ remains to be shown;
edge sides encoding $G^{(0)}_{b_kb_{l}}$ and all other pockets 
are automatic. A symbolic notation is used now to
  sketch the trees. Horizontal dots are used to indicate a general
  direct tree and horizontal dots with vertical dots above them
  indicate an opposite tree. Unspecified threads are indicated by
  dotted half-edges. The four trees mentioned above are depicted as
\begin{align*}
\setlength{\unitlength}{1.25mm}
\parbox{60mm}{\begin{picture}(40,22)(-5,-1)
\put(-5,18){\mbox{$\text{OT}_{\tilde{e}}=$}}
\put(10.3,15){\textbullet}
\put(15.3,15){\textbullet}
\put(18.75,15.5){\ldots}
\put(20,16.5){\mbox{$\vdots$}}
\put(25.3,15){\textbullet}
\put(9,13){\mbox{\scriptsize$b_{0}$}}
\put(14,13){\mbox{\scriptsize$b_{\overline{2l+1}}$}}
\put(24,13){\mbox{\scriptsize$b_{N-2}$}}
\qbezier(11,16)(18.5,27)(26,16)
\put(-5,4){\mbox{$\text{DT}_{\tilde{e}}=$}}
\put(10.3,2){\textbullet}
\put(15.3,2){\textbullet}
\put(19.25,2.5){\ldots}
\put(25.3,2){\textbullet}
\put(5,0){\mbox{\scriptsize$b_{N-1}$}}
\put(14,0){\mbox{\scriptsize$b_{2l+1}$}}
\put(24,0){\mbox{\scriptsize$b_{\overline{N-2}}$}}
\qbezier(11,3)(13.5,7)(16,3)
\thicklines
\qbezier[4](11,3)(11,6)(14,7)
\qbezier[4](16,3)(16,6)(19,7)
\qbezier[4](23,7)(26,6)(26,3)
\end{picture}}
\parbox{60mm}{\setlength{\unitlength}{1.25mm}\begin{picture}(40,22)(-5,-1)
\put(-6,5){\mbox{$\text{DT}_{\tilde{f}}=$}}
\put(10.3,2){\textbullet}
\put(16.75,2.5){\ldots}
\put(25.3,2){\textbullet}
\put(9,0){\mbox{\scriptsize$b_{1}$}}
\put(24,0){\mbox{\scriptsize$b_{\overline{2l}}$}}
\thicklines
\qbezier[4](11,3)(11,6)(14,7)
\qbezier[4](23,7)(26,6)(26,3)
\thinlines 
\put(-6,17){\mbox{$\text{OT}_{\tilde{f}}=$}}
\put(10.3,15){\textbullet}
\put(16.75,15.5){\ldots}
\put(18.2,16.5){\mbox{$\vdots$}}
\put(25.3,15){\textbullet}
\put(9,13){\mbox{\scriptsize$b_{\overline{1}}$}}
\put(24,13){\mbox{\scriptsize$b_{2l}$}}
\qbezier(11,16)(18.5,27)(26,16)
\end{picture}}
\end{align*}
Here $\tilde{e}$ describes $P_{1}$, at odd distance, so that
even-labelled nodes are connected by the opposite tree. Every edge in
the pocket tree has two sides labelled $b_{r}$ and $b_{s}$, where the
convention of Remark~\ref{rmk:sidelabels} is used when the other side
label does not matter.

The first edge in the pocket tree has side labels
$b_{0}b_{N-1}$ and descends from the root pocket. The following edge
is $b_{\overline{2l+1}}b_{2l+1}$ where $2l+2\leq \overline{2l+1} \leq
N-2$ is an even number. The
final edge is $b_{N-2}b_{\overline{N-2}}$ where $2l+1\leq
\overline{N-2} \leq N-3$ is an odd number.

Next, $\tilde{f}$ encodes $P_{0}$ in the pocket tree belonging to
$[T_{l}]_{b_{1}\ldots b_{2l}}$. It lies at even distance, but, because
the labels at $G^{(0)}_{b_{1}\ldots b_{2l}}$ start with an odd one,
the odd nodes of $\tilde{f}$ are connected by the direct tree and the
even nodes by the opposite tree.  Again, $2\leq \overline{1}\leq 2l $
denotes an even number and $1\leq \overline{2l}\leq 2l-1 $ an odd
number. When pasting $\tilde{f}$ into $\tilde{e}$, the first edge
remains $b_{0}b_{N-1}$, which descends from the root. Then all edges
from $\tilde{f}$ follow and, finally, the remaining edges of
$\tilde{e}$. Thus, before taking the denominators into account, the
four trees are arranged as:
\begin{align}
\setlength{\unitlength}{1.25mm}
\parbox{100mm}{\begin{picture}(40,28)(-5,-1)
\put(0,19){\mbox{$\text{OT}_{\tilde{e}} {\cup} 
\text{OT}_{\tilde{f}}$:}}
\put(20.3,17){\textbullet}
\put(19,15){\mbox{\scriptsize$b_{0}$}}
\put(25.3,17){\textbullet}
\put(31.75,17.5){\ldots}
\put(33,18.5){\mbox{$\vdots$}}
\put(40.55,17){\textbullet}
\put(45.3,17){\textbullet}
\put(24,15){\mbox{\scriptsize$b_{\overline{1}}$}}
\put(39,15){\mbox{\scriptsize$b_{2l}$}}
\put(45.3,17){\textbullet}
\put(49,17.5){\ldots}
\put(50.25,18.5){\mbox{$\vdots$}}
\put(45,15){\mbox{\scriptsize$b_{\overline{2l+1}}$}}
\put(55.55,17){\textbullet}
\put(54,15){\mbox{\scriptsize$b_{N-2}$}}
\qbezier(26,18)(33.5,27)(41,18)
\qbezier(21,18)(38.5,35)(56,18)
\put(0,5){\mbox{$\text{DT}_{\tilde{e}} {\cup} 
\text{DT}_{\tilde{f}}$:}}
\put(20.3,2){\textbullet}
\put(16,0){\mbox{\scriptsize$b_{N-1}$}}
\put(25.3,2){\textbullet}
\put(25,0){\mbox{\scriptsize$b_{1}$}}
\put(31.25,2.5){\ldots}
 \put(40.3,2){\textbullet}
 \put(45.3,2){\textbullet}
 \put(55.3,2){\textbullet}
 \put(44,0){\mbox{\scriptsize$b_{2l+1}$}}
 \put(39,0){\mbox{\scriptsize$b_{\overline{2l}}$}}
 \put(55,0){\mbox{\scriptsize$b_{\overline{N-2}}$}}
\qbezier(21,3)(33.5,18)(46,3)
 \thicklines
\qbezier[4](21,3)(21,6)(24,8)
\qbezier[4](26,3)(26,6)(29,7)
\qbezier[4](38,7)(41,6)(41,3)
\qbezier[4](53,7)(56,6)(56,3)
\thinlines 
\end{picture}}
\end{align}
The denominator of
$\frac{1}{(E_{b_{0}}-E_{b_{2l}})(E_{b_{N-1}}-E_{b_{1}})}$ (with
rearranged sign) corresponds to a thread between the nodes $b_{0}$ and
$b_{2l}$ and one between the nodes $b_{N-1}$ and $b_{1}$: 
\begin{align}
\setlength{\unitlength}{1.25mm}
\parbox{100mm}{\begin{picture}(40,28)(-5,-1)
\put(0,19){\mbox{$\text{OT}_{\tilde{e} \bullet \tilde{f}}$:}}
\put(20.3,17){\textbullet}
\put(19,15){\mbox{\scriptsize$b_{0}$}}
\put(25.3,17){\textbullet}
\put(31.75,17.5){\ldots}
\put(33,18.5){\mbox{$\vdots$}}
\put(40.55,17){\textbullet}
\put(45.4,17){\textbullet}
\put(24,15){\mbox{\scriptsize$b_{\overline{1}}$}}
\put(39,15){\mbox{\scriptsize$b_{2l}$}}
\put(45.3,17){\textbullet}
\put(49,17.5){\ldots}
\put(50.25,18.5){\mbox{$\vdots$}}
\put(45,15){\mbox{\scriptsize$b_{\overline{2l+1}}$}}
\put(55.55,17){\textbullet}
\put(54,15){\mbox{\scriptsize$b_{N-2}$}}
\qbezier(26,18)(33.5,27)(41,18)
\qbezier(21,18)(38.5,35)(56,18)
\qbezier(21,18)(38.5,30)(41,18)
\put(0,5){\mbox{$\text{DT}_{\tilde{e}\bullet \tilde{f}}$:}}
\put(20.3,2){\textbullet}
\put(16,0){\mbox{\scriptsize$b_{N-1}$}}
\put(25.3,2){\textbullet}
\put(25,0){\mbox{\scriptsize$b_{1}$}}
\put(31.25,2.5){\ldots}
 \put(40.3,2){\textbullet}
 \put(45.3,2){\textbullet}
 \put(55.3,2){\textbullet}
 \put(44,0){\mbox{\scriptsize$b_{2l+1}$}}
 \put(39,0){\mbox{\scriptsize$b_{\overline{2l}}$}}
 \put(55,0){\mbox{\scriptsize$b_{\overline{N-2}}$}}
\qbezier(21,3)(33.5,18)(46,3)
\qbezier(21,3)(23.5,5)(26,3)
 \thicklines
\qbezier[4](21,3)(21,6)(24,8)
\qbezier[4](26,3)(26,6)(29,7)
\qbezier[4](38,7)(41,6)(41,3)
\qbezier[4](53,7)(56,6)(56,3)
\thinlines 
\end{picture}}
\end{align}
The result
is precisely described by
$\tilde{e}\bullet\tilde{f}=(e_{0}+1,f_{0},\ldots,f_{q},e_{1},\ldots,e_{p})$
with Definitions \ref{dfnt:rpt} and \ref{dfnt:ot}. Indeed, the
increased zeroth entry corresponds to one additional half-thread
attached to the first node $b_{N-1}$ and one additional half-thread to
$b_{0}$. For the direct tree the rules imply that the next node,
$b_{1}$, is connected to $b_{N-1}$. This is the new thread from the
denominators. The next operations are done within $\tilde{f}$,
labelled $b_{1},\ldots,b_{\overline{2l}}$, without any
change. Arriving at its final node $b_{\overline{2l}}$ all
half-threads of $\tilde{f}$ are connected. The next node, labelled
$b_{2l+1}$, connects to the previous open half-thread, which is the
very first node $b_{N-1}$. These and all the following connections
arise within $\tilde{e}$ and remain unchanged. Similarly, in the
opposite tree, we first open $e_{0}+1$ half-threads at the zeroth node
$b_{0}$. Since $f_{0}>0$, we subsequently open $f_{0}$ half-threads at
the first node $b_{\overline{1}}$. The next operations remain
unchanged, until we arrive at the final node $b_{2l}$ of
$\tilde{f}$. It corresponds to $f_{q}=0$, so that we connect it to all
previous open half-threads, first within $\tilde{f}$. However, because
$e_{0}+1>0$, it is connected by an additional thread to $b_{0}$ and
encodes the denominator of $\frac{1}{E_{b_{0}}-E_{b_{2l}}}$. This
consumes the additional half-thread attached to $b_{0}$. All further
connections are the same as within $\tilde{e}$. In conclusion, we
obtain precisely the nested Catalan table
$T_{k+1}=\langle(0),\tilde{e}^{(1)}\ldots \tilde{e}^{(N/2)}\rangle$ we
started with.

\bigskip

\noindent {\bf Case II:} Let $T_{k+1}= \langle
\tilde{e}^{(0)},(0),\tilde{e}^{(2)},\ldots,\tilde{e}^{(k+1)}\rangle
\in \mathcal{T}_{k+1}$ and
$N/2=k+1$. \\
There are uniquely defined nested Catalan tables $T_{l}= \langle
\tilde{e},(0),\tilde{e}^{(2)},\ldots,\tilde{e}^{(l)}\rangle
\in \mathcal{T}_{l}$ and
$T_{k-l+1}= \langle
\tilde{f},\tilde{e}^{(l+1)},\ldots,\tilde{e}^{(k+1)}\rangle
\in \mathcal{T}_{k-l+1}$  with
$\tilde{e}^{(0)}=\tilde{e}\circ \tilde{f}$ and, consequently, $T_{l}
\smalllozenge T_{k-l+1}=T_{k+1}$. The length $l=\hat{k}$ is 
obtained via (\ref{k-lozenge}). Recall that $T_{k+1}$ cannot be
obtained by the $\smallblacklozenge$-composition, because the 1$^{\text{st}}$
entry has length $|(0)|=0$. By the induction hypothesis, $T_{l}$
encodes a unique contribution $[T_{l}]_{b_{0}\ldots b_{2l-1}}$ to
$G^{(0)}_{b_{0}\ldots b_{2l-1}}$ and $T_{k-l+1}$ encodes a unique
contribution $[T_{k-l+1}]_{b_{2l}\ldots b_{N-1}}$ to
$G^{(0)}_{b_{2l}\ldots b_{N-1}}$. It remains to be shown that
\begin{equation*}
\frac{[T_{l}]_{b_{0}\ldots b_{2l-1}}[T_{k-l+1}]_{b_{2l}\ldots
    b_{N-1}}}{(E_{b_{0}}-E_{b_{2l}})(E_{b_{1}}-E_{b_{N-1}})}
\end{equation*}
agrees with $[T_{k+1}]_{b_{0}\ldots b_{N-1}}$ encoded by $T_{k+1}$.
A detail of the pocket tree of $T_{k+1}$ 
sketching $P_0,P_1$ and their attached edges is
\begin{align}
\parbox{80mm}{\setlength{\unitlength}{1.25mm}\begin{picture}(60,12)
\put(31,12){\circle*{2}}
\put(27,13){\mbox{\scriptsize$P_{0}$}}
\put(31,12){\line(-4,-1){16}}
\put(31,12){\line(-1,-1){10}}
\put(31,12){\line(1,-6){2}}
\put(31,12){\line(4,-3){12}}
\put(31,12){\line(6,-1){20}}
\put(15,8){\circle*{2}}
\put(11,7){\mbox{\scriptsize$P_{1}$}}
\put(19,10){\mbox{\scriptsize$b_{0}$}}
\put(19,7){\mbox{\scriptsize$b_{1}$}}
\put(23,6.5){\mbox{\scriptsize$b_{2}$}}
\put(25,4){\mbox{\scriptsize$b_{\overline{2}}$}}
\put(27,0){\mbox{\scriptsize$b_{\overline{2l{-}1}}$}}
\put(33,2){\mbox{\scriptsize$b_{2l{-}1}$}}
\put(35,6){\mbox{\scriptsize$b_{2l}$}}
\put(40.5,5){\mbox{\scriptsize$b_{\overline{2l}}$}}
\put(42,8){\mbox{\scriptsize$b_{\overline{N{-}1}}$}}
\put(40,11){\mbox{\scriptsize$b_{N{-}1}$}}
\thicklines 
\qbezier[3](39,7.5)(40.2,8.2)(40.8,9.8)
\qbezier[3](28,8)(29,7)(31,7)
\end{picture}}
\end{align}
As in {\bf Case I} only the gluing of the direct and opposite tree
encoded by $\tilde{e}=(e_{0},\ldots,e_{p})$ with the direct and
opposite tree encoded by $\tilde{f}=(f_{0},\ldots,f_{q})$ via a thread
from $b_{0}$ to $b_{2l}$ and a thread from $b_{1}$ to $b_{N-1}$ must
be demonstrated. Everything else is automatic. These trees are
\begin{align}
\setlength{\unitlength}{1.25mm}\hspace*{7mm}
\parbox{55mm}{\begin{picture}(40,20)
\put(-5,5){\mbox{$\text{OT}_{\tilde{e}}=$}}
\put(10.3,2){\textbullet}
\put(15.3,2){\textbullet}
\put(19,2.5){\ldots}
\put(20.25,3.5){\mbox{$\vdots$}}
\put(25.3,2){\textbullet}
\put(9,0){\mbox{\scriptsize$b_{1}$}}
\put(14,0){\mbox{\scriptsize$b_{\overline{2}}$}}
\put(24,0){\mbox{\scriptsize$b_{2l-1}$}}
\qbezier(11,3)(18.5,14)(26,3)
\put(-5,17){\mbox{$\text{DT}_{\tilde{e}}=$}}
\put(10.3,15){\textbullet}
\put(15.55,15){\textbullet}
\put(19,15.5){\ldots}
\put(25.3,15){\textbullet}
\put(9,13){\mbox{\scriptsize$b_{0}$}}
\put(15,13){\mbox{\scriptsize$b_{2}$}}
\put(24,13){\mbox{\scriptsize$b_{\overline{2l-1}}$}}
\qbezier(11,16)(13.5,20)(16,16)
\thicklines
\qbezier[4](11,16)(11,19)(14,20)
\qbezier[4](16,16)(16,19)(19,20)
\qbezier[4](23,20)(26,19)(26,16)
\end{picture}}
\parbox{50mm}{\setlength{\unitlength}{1.25mm}\begin{picture}(40,25)(-5,-1)
\put(-6,17){\mbox{$\text{DT}_{\tilde{f}}=$}}
\put(10.3,15){\textbullet}
\put(16.75,15.5){\ldots}
\put(25.3,15){\textbullet}
\put(9,13){\mbox{\scriptsize$b_{2l}$}}
\put(24,13){\mbox{\scriptsize$b_{\overline{N-1}}$}}
\thicklines
\qbezier[4](11,16)(11,19)(14,20)
\qbezier[4](23,20)(26,19)(26,16)
\thinlines 
\put(-6,5){\mbox{$\text{OT}_{\tilde{f}}=$}}
\put(10.3,2){\textbullet}
\put(16.75,2.5){\ldots}
\put(18,3.5){\mbox{$\vdots$}}
\put(25.3,2){\textbullet}
\put(9,0){\mbox{\scriptsize$b_{\overline{2l}}$}}
\put(24,0){\mbox{\scriptsize$b_{N-1}$}}
\qbezier(11,3)(18.5,14)(26,3)
\end{picture}}
\label{e:case-2sep}
\end{align}
The notation is the same as in {\bf Case I}. The 1$^{\text{st}}$ pocket
$P_{1}$, described by the Catalan tuple $(0)$, is only $1$-valent so
that the first edge is labelled $b_{0}b_{1}$. The direct trees in
(\ref{e:case-2sep}) are put next to each other and a thread between
$b_{0}$ and $b_{2l}$ is drawn for the denominator of
$\frac{1}{E_{b_{0}}-E_{b_{2l}}}$. Similarly, the opposite trees in
(\ref{e:case-2sep}) are put next to each other and a thread between
$b_{1}$ and $b_{N-1}$ is drawn for the denominator of
$\frac{1}{E_{b_{1}}-E_{b_{N-1}}}$:
\begin{align*}
\parbox{100mm}{\setlength{\unitlength}{1.25mm}\begin{picture}(40,25)(-5,-1)
\put(-5,5){\mbox{$\text{OT}_{\tilde{e}\circ \tilde{f}}=$}}
\put(15.55,2){\textbullet}
\put(20.3,2){\textbullet}
\put(23.75,2.5){\ldots}
\put(25,3.5){\mbox{$\vdots$}}
\put(30.3,2){\textbullet}
\put(14,0){\mbox{\scriptsize$b_{1}$}}
\put(19,0){\mbox{\scriptsize$b_{\overline{2}}$}}
\put(29,0){\mbox{\scriptsize$b_{2l-1}$}}
\qbezier(16,3)(23.5,14)(31,3)
\put(40.3,2){\textbullet}
\put(46.75,2.5){\ldots}
\put(48,3.5){\mbox{$\vdots$}}
\put(55.8,2){\textbullet}
\put(39,0){\mbox{\scriptsize$b_{\overline{2l}}$}}
\put(54,0){\mbox{\scriptsize$b_{N-1}$}}
\qbezier(41,3)(48.5,14)(56,3)
\put(28.5,16){\oval(25,13)[t]}
\put(-5,17){\mbox{$\text{DT}_{\tilde{e}\circ\tilde{f}}=$}}
\put(15.55,15){\textbullet}
\put(20.55,15){\textbullet}
\put(23.75,15.5){\ldots}
\put(30.3,15){\textbullet}
\put(14,13){\mbox{\scriptsize$b_{0}$}}
\put(20,13){\mbox{\scriptsize$b_{2}$}}
\put(29,13){\mbox{\scriptsize$b_{\overline{2l-1}}$}}
\qbezier(16,16)(18.5,20)(21,16)
\thicklines
\qbezier[4](16,16)(16,19)(19,20)
\qbezier[4](21,16)(21,19)(24,20)
\qbezier[4](28,20)(31,19)(31,16)
\put(40.8,15){\textbullet}
\put(46.75,15.5){\ldots}
\put(55.8,15){\textbullet}
\put(39,13){\mbox{\scriptsize$b_{2l}$}}
\put(54,13){\mbox{\scriptsize$b_{\overline{N-1}}$}}
\thicklines
\qbezier[4](41,16)(41,19)(44,20)
\qbezier[4](53,20)(56,19)(56,16)
\thinlines 
\put(36,3){\oval(40,15)[t]}
\end{picture}}
\end{align*}
The result are precisely the direct and opposite trees of
the composition
$\tilde{e}\circ\tilde{f}=(e_{0}+1,e_{1},\ldots,e_{p},f_{0},\ldots,f_{q})$. The
increase $e_{0}\rightarrow e_{0}+1$ opens an additional half-thread at
$b_{0}$ and an additional half-thread at $b_{1}$. In the direct tree,
this new half-thread is not used by $e_{1},\ldots,e_{p}$. Only when we
are moving to $f_{0}$, labelled $b_{2l}$, we have to connect it with
the last open half-thread, i.e.\ with $b_{0}$. After that the remaining
operations are unchanged compared with $\tilde{f}$. In the opposite
tree, the additional half-thread at $b_{1}$ is not used in
$e_{1},\ldots,e_{p}$. Because $f_{0}$, labelled
$b_{\overline{2l}}$, opens enough half-threads, it is not consumed by
$f_{0},\ldots,f_{q-1}$ either. Then, the last node $f_{q}$, labelled
$b_{N-1}$, successively connects to all nodes with open half-threads,
including $b_{1}$. In conclusion, we obtain precisely the nested Catalan
table $T_{k+1}=\langle\tilde{e}^{(0)},(0),\tilde{e}^{(2)}\ldots
\tilde{e}^{(N/2)}\rangle$ we started with.

\bigskip

\noindent {\bf Case III:} Finally, we consider a general
$T_{k+1}=\langle\tilde{e}^{(0)},\tilde{e}^{(1)},\tilde{e}^{(2)},\ldots
\tilde{e}^{(k+1)}\rangle\in\mathcal{T}_{k+1}$ with $k+1=N/2$, $|\tilde{e}^{(0)}|\geq1 $
and $|\tilde{e}^{(1)}|\geq 1$.
There are uniquely defined nested Catalan tables $T_{l}= \langle
\tilde{e},\tilde{e}^{(1)},\tilde{e}^{(2)},\ldots,\tilde{e}^{(l)}\rangle\in\mathcal{T}_l$
and $T_{k-l+1}= \langle
\tilde{f},\tilde{e}^{(l+1)},\ldots,\tilde{e}^{(k+1)}\rangle\in\mathcal{T}_{k-l+1}$ with
$\tilde{e}^{(0)}=\tilde{e}\circ \tilde{f}$ and consequently $T_{l}
\smalllozenge T_{k-l+1}=T_{k+1}$. Moreover, uniquely defined nested Catalan
tables $T_{l'}= \langle
\tilde{f}',\tilde{e}^{(2)},\ldots,\tilde{e}^{(l'+1)}\rangle\in\mathcal{T}_{l'}$ and
$T_{k-l'+1}= \langle
\tilde{e}^{(0)},\tilde{e}',\tilde{e}^{(l'+2)},\ldots,\tilde{e}^{(k+1)}\rangle\in\mathcal{T}_{k-l'+1}$
exist, such that $\tilde{e}^{(1)}=\tilde{e}'\bullet \tilde{f}'$ and
consequently $T_{k-l'+1}\smallblacklozenge T_{l'}=T_{k+1}$. We
necessarily have $l'\leq k-1$ and $l\geq 2$, because $l'=k$
corresponds to {\bf Case I} and $l=1$ to {\bf Case II}. By the
induction hypothesis, these nested Catalan subtables encode unique
contributions $[T_{l}]_{b_{0}\ldots b_{2l-1}}$ to
$G^{(0)}_{b_{0}\ldots b_{2l-1}}$, $[T_{k-l+1}]_{b_{2l}\ldots
  b_{N-1}}$ to $G^{(0)}_{b_{2l}\ldots b_{N-1}}$,
$[T_{l'}]_{b_{1}\ldots b_{2l'}}$ to $G^{(0)}_{b_{1}\ldots b_{2l'}}$
and $[T_{k-l'+1}]_{b_{0}b_{2l'+1}\ldots b_{N-1}}$ to
$G^{(0)}_{b_{0}b_{2l'+1}\ldots b_{N-1}}$. We have to show that
\begin{align}
&\hspace{-8mm}\frac{[T_{l}]_{b_{0}\ldots b_{2l-1}}[T_{N/2-l}]_{b_{2l}\ldots b_{N-1}}}{(E_{b_{0}}-E_{b_{2l}})(E_{b_{1}}-E_{b_{N-1}})}-\frac{[T_{l'}]_{b_{1}\ldots b_{2l'}}[T_{N/2-l'}]_{b_{0}b_{2l'+1}\ldots b_{N-1}}}{(E_{b_{0}}-E_{b_{2l'}})(E_{b_{1}}-E_{b_{N-1}})}\label{e:case3-G}
\end{align}
agrees with $[T_{k+1}]_{b_{0}\ldots,b_{N-1}}$. 

\smallskip

In the pocket tree of $T_{k+1}$ there must be an edge with side labels 
$b_{0}b_{h}$, where $3\leq h\leq N-3$ and $h$ is odd. Here is a detail
of the pocket tree of $T_{k+1}$ showing $P_0,P_1$:
\begin{align}
\parbox{80mm}{\setlength{\unitlength}{1.25mm}\begin{picture}(60,17)
\put(31,17){\circle*{2}}
\put(27,18){\mbox{\scriptsize$P_{0}$}}
\put(31,17){\line(-4,-1){24}}
\put(31,17){\line(-1,-1){10}}
\put(31,17){\line(1,-6){2}}
\put(31,17){\line(4,-3){12}}
\put(31,17){\line(6,-1){20}}
\put(7,11){\circle*{2}}
\put(3,12){\mbox{\scriptsize$P_{1}$}}
\put(7,11){\line(-3,-2){10}}
\put(7,11){\line(-1,-3){3.5}}
\put(7,11){\line(1,-2){6.5}}
\put(7,11){\line(2,-1){12}}
\put(12,13.5){\mbox{\scriptsize$b_{0}$}}
\put(15,11.5){\mbox{\scriptsize$b_{h}$}}
\put(-2,7){\mbox{\scriptsize$b_{1}$}}
\put(-2,3){\mbox{\scriptsize$b_{\overline{1}}$}}
\put(12,8.5){\mbox{\scriptsize$b_{h{-}1}$}}
\put(13,4.5){\mbox{\scriptsize$b_{\overline{h{-}1}}$}}
\put(20,11.5){\mbox{\scriptsize$b_{h{+}1}$}}
\put(25,9){\mbox{\scriptsize$b_{\overline{h{+}1}}$}}
\put(27,5){\mbox{\scriptsize$b_{\overline{2l{-}1}}$}}
\put(33,7){\mbox{\scriptsize$b_{2l{-}1}$}}
\put(35,11){\mbox{\scriptsize$b_{2l}$}}
\put(40.5,10){\mbox{\scriptsize$b_{\overline{2l}}$}}
\put(42,13){\mbox{\scriptsize$b_{\overline{N{-}1}}$}}
\put(40,16){\mbox{\scriptsize$b_{N{-}1}$}}
\put(3.5,-1){\mbox{\scriptsize$b_{2l'}$}}
\put(5.8,1.5){\mbox{\scriptsize$b_{2l'\!{+}1}$}}
\put(13,-0.5){\mbox{\scriptsize$b_{\overline{2l'{+}1}}$}}
\put(1,4){\mbox{\scriptsize$b_{\overline{2l'}}$}}
\thicklines 
\qbezier[3](39,12.5)(40.2,13.2)(40.8,14.8)
\qbezier[3](28,13)(29,12)(31,12)
\qbezier[3](2,7)(2.8,6)(4.6,5.5)
\qbezier[3](10.5,5.2)(11.8,5.5)(12.8,7.3)
\end{picture}}
\end{align}
The direct and
opposite trees for $\tilde{e},\tilde{f}$ and $\tilde{e}^{(1)}$ can 
be sketched as
\begin{align}
\parbox{90mm}{\setlength{\unitlength}{1.25mm}\hspace*{10mm}
\begin{picture}(60,25)(8,-1)
\put(0,5){\mbox{$\text{OT}_{\tilde{e}}{\cup}\text{OT}_{\tilde{f}}=$}}
\put(20.3,2){\textbullet}
\put(25.3,2){\textbullet}
\put(29,2.5){\ldots}
\put(30.25,3.5){\mbox{$\vdots$}}
\put(35.3,2){\textbullet}
\put(19,0){\mbox{\scriptsize$b_{h}$}}
\put(24,0){\mbox{\scriptsize$b_{\overline{h+1}}$}}
\put(34,0){\mbox{\scriptsize$b_{2l-1}$}}
\qbezier(21,3)(28.5,14)(36,3)
\put(45.3,2){\textbullet}
\put(51.75,2.5){\ldots}
\put(53.1,3.5){\mbox{$\vdots$}}
\put(60.3,2){\textbullet}
\put(44,0){\mbox{\scriptsize$b_{\overline{2l}}$}}
\put(59,0){\mbox{\scriptsize$b_{N-1}$}}
\qbezier(46,3)(53.5,14)(61,3)
\put(45.3,15){\textbullet}
\put(51.57,15.5){\ldots}
\put(60.3,15){\textbullet}
\put(44,13){\mbox{\scriptsize$b_{2l}$}}
\put(59,13){\mbox{\scriptsize$b_{\overline{N-1}}$}}
\thicklines
\qbezier[4](46,16)(46,19)(49,20)
\qbezier[4](58,20)(61,19)(61,16)
\thinlines 
\put(0,17){\mbox{$\text{DT}_{\tilde{e}}{\cup}\text{DT}_{\tilde{f}}=$}}
\put(20.3,15){\textbullet}
\put(25.3,15){\textbullet}
\put(29,15.5){\ldots}
\put(35.3,15){\textbullet}
\put(19,13){\mbox{\scriptsize$b_{0}$}}
\put(24,13){\mbox{\scriptsize$b_{h+1}$}}
\put(34,13){\mbox{\scriptsize$b_{\overline{2l-1}}$}}
\qbezier(21,16)(23.5,20)(26,16)
\thicklines
\qbezier[4](21,16)(21,19)(24,20)
\qbezier[4](26,16)(26,19)(29,20)
\qbezier[4](33,20)(36,19)(36,16)
\end{picture}}\quad
\parbox{36mm}{\setlength{\unitlength}{1.25mm}\begin{picture}(20,25)(8,-1)
\put(3,17){\mbox{$\text{OT}_{\tilde{e}^{(1)}}=$}}
\put(15.3,15){\textbullet}
\put(20.3,15){\textbullet}
\put(23.75,15.5){\ldots}
\put(25,16.5){\mbox{$\vdots$}}
\put(30.3,15){\textbullet}
\put(14,13){\mbox{\scriptsize$b_{0}$}}
\put(20,13){\mbox{\scriptsize$b_{\overline{1}}$}}
\put(29,13){\mbox{\scriptsize$b_{h-1}$}}
\qbezier(16,16)(23.5,27)(31,16)
\put(3,5){\mbox{$\text{DT}_{\tilde{e}^{(1)}}=$}}
\put(15.3,2){\textbullet}
\put(20.3,2){\textbullet}
\put(23.75,2.5){\ldots}
\put(30.3,2){\textbullet}
\put(14,0){\mbox{\scriptsize$b_{h}$}}
\put(20,0){\mbox{\scriptsize$b_{1}$}}
\put(29,0){\mbox{\scriptsize$b_{\overline{h-1}}$}}
\qbezier(16,3)(18.5,7)(21,3)
\thicklines
\qbezier[4](16,3)(16,6)(19,7)
\qbezier[4](21,3)(21,6)(24,7)
\qbezier[4](28,7)(31,6)(31,3)
\end{picture}\rule[-2mm]{0mm}{1cm}}
\label{e:case-3sepa}
\end{align}
The denominators of
$\frac{1}{(E_{b_{0}}-E_{b_{2l}})(E_{b_{1}}-E_{b_{N-1}})}$ in
(\ref{e:case3-G}) add threads from $b_{0}$ to $b_{2l}$ and from
$b_{1}$ to $b_{N-1}$. The first one connects the direct trees for
$\tilde{e}\cup \tilde{f}$ to the direct tree encoded by
$\tilde{e}^{(0)}=\tilde{e}\circ \tilde{f}$. The second thread does \emph{not}
give a valid composition of the opposite trees for $\tilde{e}\cup
\tilde{f}$.

This is a problem. The solution is to split this
contribution. Half of the contribution is sacrificed to bring the
other half in the desired form. Afterwards, the same procedure is
repeated for the other term in (\ref{e:case3-G}) with a
minus-sign. The remainders are the same and cancel each other, whereas
the other halves add up to yield the sought for monomial.

Returning to trees, we note that in the direct tree for the pocket
$\tilde{e}^{(1)}$ there is always a thread from $b_{h}$ to $b_{1}$,
encoding a factor $\frac{1}{E_{b_{h}}-E_{b_{1}}}$. With the factor
$\frac{1}{E_{b_{1}}-E_{b_{N-1}}}$ it fulfils
\begin{equation}
\frac{1}{E_{b_{h}}{-}E_{b_{1}}}\cdot \frac{1}{E_{b_{1}}{-}E_{b_{N-1}}}
= \frac{1}{E_{b_{h}}{-}E_{b_{1}}}\cdot
\frac{1}{E_{b_{h}}{-}E_{b_{N-1}}}
+\frac{1}{E_{b_{h}}{-}E_{b_{N-1}}}\cdot
\frac{1}{E_{b_{1}}{-}E_{b_{N-1}}}\;.
\label{e:ident-E}
\end{equation}
The first term on the right-hand side of (\ref{e:ident-E}) leaves the
direct tree $\text{DT}_{\tilde{e}^{(1)}}$ as it is and connects the
parts of $\text{OT}_{\tilde{e}}{\cup}\text{OT}_{\tilde{f}}$ via the
thread from $b_{h}$ to $b_{N-1}$ to form
$\text{OT}_{\tilde{e}^{(0)}}$, where $\tilde{e}^{(0)}=\tilde{e}\circ\tilde{f}$.

\smallskip

\noindent
The final term in (\ref{e:ident-E}) also unites
$\text{OT}_{\tilde{e}}{\cup}\text{OT}_{\tilde{f}}$ and forms
$\text{OT}_{\tilde{e}^{(0)}}$, but it removes in
$\text{DT}_{\tilde{e}^{(1)}}$ the thread between $b_{h}$ and
$b_{1}$. It follows from $\tilde{e}^{(1)}=\tilde{e}'\bullet
\tilde{f}'$ that this tree falls apart into the subtrees
$\text{DT}_{\tilde{e}'}$, containing $b_{h}$, and
$\text{DT}_{\tilde{f}'}$, which contains $b_{1}$. These are multiplied
by a factor $\frac{1}{E_{b_{1}}-E_{b_{N-1}}}$. The second term in
(\ref{e:case3-G}) will remove them.

Indeed, direct and opposite trees for $\tilde{e}^{(0)},\tilde{e}'$ 
and $\tilde{f}'$ can be sketched as 
\begin{align}
\parbox{50mm}{\setlength{\unitlength}{1.25mm}\begin{picture}(30,25)(7,-1)
\put(2,5){\mbox{$\text{OT}_{\tilde{e}^{(0)}}=$}}
\put(15.3,2){\textbullet}
\put(20.3,2){\textbullet}
\put(23.75,2.5){\ldots}
\put(25.2,3.5){\mbox{$\vdots$}}
\put(30.3,2){\textbullet}
\put(14,0){\mbox{\scriptsize$b_{h}$}}
\put(19,0){\mbox{\scriptsize$b_{\overline{h+1}}$}}
\put(29,0){\mbox{\scriptsize$b_{N-1}$}}
\qbezier(16,3)(23.5,14)(31,3)
\put(2,17){\mbox{$\text{DT}_{\tilde{e}^{(0)}}=$}}
\put(15.3,15){\textbullet}
\put(20.3,15){\textbullet}
\put(23.75,15.5){\ldots}
\put(30,15){\textbullet}
\put(15,13){\mbox{\scriptsize$b_{0}$}}
\put(19,13){\mbox{\scriptsize$b_{h+1}$}}
\put(29,13){\mbox{\scriptsize$b_{\overline{N-1}}$}}
\qbezier(16,16)(18.5,20)(21,16)
\thicklines
\qbezier[4](16,16)(16,19)(19,20)
\qbezier[4](21,16)(21,19)(24,20)
\qbezier[4](28,20)(31,19)(31,16)
\end{picture}}\quad
\parbox{60mm}{\setlength{\unitlength}{1.25mm}\begin{picture}(40,25)(14,-1)
\put(5,17){\mbox{$\text{OT}_{\tilde{e}'}
{\cup} \text{OT}_{\tilde{f}'}\,{=}$}}
\put(25.3,15){\textbullet}
\put(30.3,15){\textbullet}
\put(33.75,15.5){\ldots}
\put(35.2,16.3){\mbox{$\vdots$}}
\put(40.3,15){\textbullet}
\put(50.3,15){\textbullet}
\put(60.3,15){\textbullet}
\put(53.75,15.5){\ldots}
\put(55.3,16.5){\mbox{$\vdots$}}
\put(24,13){\mbox{\scriptsize$b_{0}$}}
\put(30,13){\mbox{\scriptsize$b_{\overline{1}}$}}
\put(39,13){\mbox{\scriptsize$b_{2l'}$}}
\put(47,13){\mbox{\scriptsize$b_{\overline{2l'+1}}$}}
\put(59,13){\mbox{\scriptsize$b_{h-1}$}}
\qbezier(31,16)(35,24)(41,16)
\put(43.5,16){\oval(35,12)[t]}
\put(5,5){\mbox{$\text{DT}_{\tilde{e}'}
{\cup} \text{DT}_{\tilde{f}'}\,{=}$}}
\put(25.3,2){\textbullet}
\put(30.3,2){\textbullet}
\put(33.75,2.5){\ldots}
\put(40.3,2){\textbullet}
\put(50.3,2){\textbullet}
\put(60.3,2){\textbullet}
\put(24,0){\mbox{\scriptsize$b_{h}$}}
\put(30,0){\mbox{\scriptsize$b_{1}$}}
\put(39,0){\mbox{\scriptsize$b_{\overline{2l}}$}}
\put(47,0){\mbox{\scriptsize$b_{2l'+1}$}}
\put(59,0){\mbox{\scriptsize$b_{\overline{h-1}}$}}
\put(38.5,3){\oval(25,12)[t]}
\thicklines
\qbezier[4](26,3)(26,6)(27,8)
\qbezier[4](31,3)(31,6)(34,7)
\qbezier[4](38,7)(41,6)(41,3)
\qbezier[4](51,3)(51,6)(54,7)
\qbezier[4](58,7)(61,6)(61,3)
\end{picture}}
\end{align}
The direct tree $\text{DT}_{\tilde{e}^{(0)}}$ remains intact and the
thread from $b_{0}$ to $b_{2l'}$ encoded in the factor
$\frac{1}{(E_{b_{0}}-E_{b_{2l'}})}$ in (\ref{e:case3-G}) connects the
opposite trees for $\tilde{e}'\cup \tilde{f}'$ to form the opposite
tree for $\tilde{e}^{(1)}=\tilde{e}'\bullet \tilde{f}'$. The direct
trees $\text{DT}_{\tilde{e}'}{\cup} \text{DT}_{\tilde{f}'}$ remain
disconnected and are multiplied by $\frac{1}{(E_{b_{1}}-E_{b_{N-1}})}$
from (\ref{e:case3-G}). With the minus-sign from (\ref{e:case3-G})
they cancel the final term in (\ref{e:ident-E}). The other trees
combined yield precisely the direct and opposite trees for both
$\tilde{e}^{(0)}$ and $\tilde{e}^{(1)}$, so that the single nested
Catalan table we started with is retrieved.

This completes the proof. Bijectivity between nested Catalan tables
and contributing terms to $(N'{<}N)$-point functions is essential:
Assuming the above construction in {\bf Cases I--III} missed nested
Catalan subtables $T_{l}, T_{N/2-l}$, then their composition
$T_{l} \smalllozenge T_{N/2-l}$ would be a new nested Catalan table of length
$N/2$. However, all nested Catalan tables of length $N/2$ are
considered. Similarly for $T_{l'} \smallblacklozenge T_{N/2-l'}$.
\end{proof}

This theorem shows that there is a one-to-one correspondence between
nested Catalan tables and the diagrams/terms in
$G^{(0)}_{b_{0}\ldots b_{N-1}}$ with designated node $b_{0}$. The
choice of designated node does not influence
$G^{(0)}_{b_{0}\ldots b_{N-1}}$, but it does alter its expansion.

\begin{appendix}

\section{Examples}
\label{app:ex}

\begin{exm}
We have $(1,0)=(0)\circ(0)$,  $(2,0,0)=(1,0)\circ(0)$,  
$(1,1,0)=(0)\circ(1,0)$ and 
$(3,1,0,0,2,0,0)=(2,1,0,0)\circ(2,0,0)$.
\end{exm}

\begin{exm}
We have $(1,0)=(0)\bullet(0)$,  $(2,0,0)=(1,0)\bullet (0)$,  
$(1,1,0)=(0)\bullet (1,0)$ and 
$(3,1,0,0,2,0,0)=(2,0,2,0,0)\bullet(1,0)$.
\end{exm}

\begin{rmk}\label{Dyck-circ}
  We formulate the $\circ$-decomposition in terms of
  Dyck paths. 
  \begin{enumerate}\item  Remove the lowest row of the lattice.
\item     Draw the
  north-east diagonal from the new bottom-left corner.  Let $F$ be the
  first step east which goes below the north-east diagonal.
\item  Remove the column containing $F$. 

\item  The left
  $\circ$-factor is the Dyck path in the lattice obtained by
  retaining only the rows and columns
  shared by the part of the north-east diagonal left of $F$.

\item The
  right $\circ$-factor is the Dyck path in the lattice obtained by
  deleting (in addition to steps 1.\ and 3.)  all rows and columns shared by
  the part of the north-east diagonal left of $F$.

\end{enumerate}
If the resulting lattice in one of the factors is empty this
  corresponds to the Catalan tuple $(0)$ of length $0$.
 For example, the decomposition $(3,1,0,0,2,0,0)=(2,1,0,0)\circ(2,0,0)$
visualises as

\centerline{$
  \parbox{32mm}{\begin{picture}(30,32)                
\put(0,0){\line(1,0){30}}
\put(0,5){\line(1,0){30}}
\put(0,10){\line(1,0){30}}
\put(0,15){\line(1,0){30}}
\put(0,20){\line(1,0){30}}
\put(0,25){\line(1,0){30}}
\put(0,30){\line(1,0){30}}
\put(0,0){\line(0,1){30}}
\put(5,0){\line(0,1){30}}
\put(10,0){\line(0,1){30}}
\put(15,0){\line(0,1){30}}
\put(20,0){\line(0,1){30}}
\put(25,0){\line(0,1){30}}
\put(30,0){\line(0,1){30}}
\linethickness{0.7 mm}
\put(0,0){\line(0,1){15}}
\put(0,15){\line(1,0){5}}
\put(5,15){\line(0,1){5}}
\put(5,20){\line(1,0){15}}
\put(20,20){\line(0,1){10}}
\put(20,30){\line(1,0){10}}
\end{picture}}
=\;\;
\parbox{32mm}{\begin{picture}(30,32)                
 \definecolor{rulecolor}{named}{lightgray}
\put(0,20){\rule{30mm}{10mm}}
 \put(15,0){\rule{15mm}{30mm}}
\definecolor{rulecolor}{named}{gray}
\put(0,0){\rule{30mm}{5mm}}
\put(15,0){\rule{5mm}{30mm}}
\multiput(0,5)(1,1){15}{.}
\put(0,0){\line(1,0){30}}  
\put(0,5){\line(1,0){30}}
\put(0,10){\line(1,0){30}}
\put(0,15){\line(1,0){30}}
\put(0,20){\line(1,0){30}}
\put(0,25){\line(1,0){30}}
\put(0,30){\line(1,0){30}}
\put(0,0){\line(0,1){30}}
\put(5,0){\line(0,1){30}}
\put(10,0){\line(0,1){30}}
\put(15,0){\line(0,1){30}}
\put(20,0){\line(0,1){30}}
\put(25,0){\line(0,1){30}}
\put(30,0){\line(0,1){30}}
\linethickness{0.7 mm}
\put(0,0){\line(0,1){15}}
\put(0,15){\line(1,0){5}}
\put(5,15){\line(0,1){5}}
\put(5,20){\line(1,0){15}}
\put(20,20){\line(0,1){10}}
\put(20,30){\line(1,0){10}}
\end{picture}}
\circ\;\;
\parbox{32mm}{\begin{picture}(30,32)                
\definecolor{rulecolor}{named}{lightgray}
\put(0,0){\rule{30mm}{20mm}}
\put(0,0){\rule{15mm}{30mm}}
\definecolor{rulecolor}{named}{gray}
\put(0,0){\rule{30mm}{5mm}}
\put(15,0){\rule{5mm}{30mm}}
\multiput(0,5)(1,1){15}{.}
\put(0,0){\line(1,0){30}}  
\put(0,5){\line(1,0){30}}
\put(0,10){\line(1,0){30}}
\put(0,15){\line(1,0){30}}
\put(0,20){\line(1,0){30}}
\put(0,25){\line(1,0){30}}
\put(0,30){\line(1,0){30}}
\put(0,0){\line(0,1){30}}
\put(5,0){\line(0,1){30}}
\put(10,0){\line(0,1){30}}
\put(15,0){\line(0,1){30}}
\put(20,0){\line(0,1){30}}
\put(25,0){\line(0,1){30}}
\put(30,0){\line(0,1){30}}
\linethickness{0.7 mm}
\put(0,0){\line(0,1){15}}
\put(0,15){\line(1,0){5}}
\put(5,15){\line(0,1){5}}
\put(5,20){\line(1,0){15}}
\put(20,20){\line(0,1){10}}
\put(20,30){\line(1,0){10}}
\end{picture}}$}

\vskip 1ex

\noindent
The row and column removed in steps 1.\ and 3.\ are shown in darker gray.
The north-east diagonal is dotted.
\end{rmk}

\begin{rmk}\label{Dyck-bullet}
  We formulate the $\bullet$-decomposition in terms of
  Dyck paths. 
  \begin{enumerate}\item  Remove the lowest row of the lattice.
\item     Draw the
  north-east diagonal from the end point of the very first step east.
  Let $F$ be the first step east which goes below this north-east diagonal.

\item  Remove the column containing $F$. 

\item The
  left $\bullet$-factor is the Dyck path in the lattice obtained by
  deleting (in addition to steps 1.\ and 3.)  all rows and columns
  shared by the part of the north-east diagonal left of $F$. 

\item  The right
  $\bullet$-factor is the Dyck path in the lattice obtained by
  retaining only the rows and columns
  shared by the part of the north-east diagonal left of $F$.
\end{enumerate}
If the resulting lattice in one of the factors is empty this
  corresponds to the Catalan tuple $(0)$ of length $0$.
  For example, the decomposition $(3,1,0,0,2,0,0)=(2,0,2,0,0)\bullet
  (1,0)$ visualises as
\\ \centerline{$
  \parbox{32mm}{\begin{picture}(30,32)                
\put(0,0){\line(1,0){30}}
\put(0,5){\line(1,0){30}}
\put(0,10){\line(1,0){30}}
\put(0,15){\line(1,0){30}}
\put(0,20){\line(1,0){30}}
\put(0,25){\line(1,0){30}}
\put(0,30){\line(1,0){30}}
\put(0,0){\line(0,1){30}}
\put(5,0){\line(0,1){30}}
\put(10,0){\line(0,1){30}}
\put(15,0){\line(0,1){30}}
\put(20,0){\line(0,1){30}}
\put(25,0){\line(0,1){30}}
\put(30,0){\line(0,1){30}}
\linethickness{0.7 mm}
\put(0,0){\line(0,1){15}}
\put(0,15){\line(1,0){5}}
\put(5,15){\line(0,1){5}}
\put(5,20){\line(1,0){15}}
\put(20,20){\line(0,1){10}}
\put(20,30){\line(1,0){10}}
\end{picture}}
=\;\;
\parbox{32mm}{\begin{picture}(30,32)                
\definecolor{rulecolor}{named}{lightgray}
\put(0,15){\rule{30mm}{5mm}}    
\put(5,0){\rule{10mm}{30mm}}
\definecolor{rulecolor}{named}{gray}
\put(0,0){\rule{30mm}{5mm}}
\put(10,0){\rule{5mm}{30mm}}
\multiput(5,15)(1,1){5}{.}
\put(0,0){\line(1,0){30}}  
\put(0,5){\line(1,0){30}}
\put(0,10){\line(1,0){30}}
\put(0,15){\line(1,0){30}}
\put(0,20){\line(1,0){30}}
\put(0,25){\line(1,0){30}}
\put(0,30){\line(1,0){30}}
\put(0,0){\line(0,1){30}}
\put(5,0){\line(0,1){30}}
\put(10,0){\line(0,1){30}}
\put(15,0){\line(0,1){30}}
\put(20,0){\line(0,1){30}}
\put(25,0){\line(0,1){30}}
\put(30,0){\line(0,1){30}}
\linethickness{0.7 mm}
\put(0,0){\line(0,1){15}}
\put(0,15){\line(1,0){5}}
\put(5,15){\line(0,1){5}}
\put(5,20){\line(1,0){15}}
\put(20,20){\line(0,1){10}}
\put(20,30){\line(1,0){10}}
\end{picture}}
\bullet\;\;
\parbox{32mm}{\begin{picture}(30,32)                
\definecolor{rulecolor}{named}{lightgray}
\put(0,0){\rule{30mm}{15mm}}
\put(0,20){\rule{30mm}{10mm}}    
\put(0,0){\rule{5mm}{30mm}}
\put(15,0){\rule{15mm}{30mm}}
\definecolor{rulecolor}{named}{gray}
\put(0,0){\rule{30mm}{5mm}}
\put(10,0){\rule{5mm}{30mm}}
\multiput(5,15)(1,1){5}{.}
\put(0,0){\line(1,0){30}}  
\put(0,5){\line(1,0){30}}
\put(0,10){\line(1,0){30}}
\put(0,15){\line(1,0){30}}
\put(0,20){\line(1,0){30}}
\put(0,25){\line(1,0){30}}
\put(0,30){\line(1,0){30}}
\put(0,0){\line(0,1){30}}
\put(5,0){\line(0,1){30}}
\put(10,0){\line(0,1){30}}
\put(15,0){\line(0,1){30}}
\put(20,0){\line(0,1){30}}
\put(25,0){\line(0,1){30}}
\put(30,0){\line(0,1){30}}
\linethickness{0.7 mm}
\put(0,0){\line(0,1){15}}
\put(0,15){\line(1,0){5}}
\put(5,15){\line(0,1){5}}
\put(5,20){\line(1,0){15}}
\put(20,20){\line(0,1){10}}
\put(20,30){\line(1,0){10}}
\end{picture}}$}

\vskip 1ex

\noindent
The row and column removed in steps 1.\ and 3.\ are shown in darker gray.
The north-east diagonal is dotted.
\end{rmk}

\begin{exm} We have 
\begin{align*}
\hspace*{-1.4cm}
\mathcal{T}_1 &=\{\langle (0),(0)\rangle \}\;, 
\\
\hspace*{-1.4cm}\mathcal{T}_2 &=\{\langle (1,0),(0),(0)\rangle ,~ \langle (0),(1,0),(0)\rangle
\}
\\
\hspace*{-1.4cm}\mathcal{T}_3 &=\{\langle (2,0,0),(0),(0),(0)\rangle ,~
\langle (1,1,0),(0),(0),(0)\rangle ,~
\langle (1,0),(1,0),(0),(0)\rangle ,
\\
&\qquad \langle (1,0),(0),(1,0),(0)\rangle ,~
\langle (0),(2,0,0),(0),(0)\rangle ,~
\langle (0),(1,1,0),(0),(0)\rangle , 
\\
&\qquad \langle (0),(1,0),(1,0),(0)\rangle \}\;.
\end{align*}
Later in Fig.~\ref{f:G4} and \ref{f:G6} we give a diagrammatic representation
of the nested Catalan tables in $\mathcal{T}_2$ and $\mathcal{T}_3$,
respectively. 
\end{exm} 

\begin{exm}
We have $\langle(2,0,0),(0),(0),(0)\rangle=
\langle(1,0),(0),(0)\rangle \smalllozenge \langle(0),(0)\rangle$ and
$\langle(1,1,0),(0),(0),(0)\rangle=\langle(0),(0)\rangle\smalllozenge 
\langle(1,0),(0),(0)\rangle$. In Ex.~\ref{ex:G12a} and 
Figure~\ref{f:G12a} 
we considered the nested Catalan table
$\langle (2,0,0),(1,1,0),(0),(0),(0),(1,0),(0)\rangle
= \langle (1,0),(1,1,0),(0),(0),(0)\rangle 
\smalllozenge \langle (0),(1,0),(0)\rangle$.
Another example will be given in Ex.~\ref{exm:triangle}.
\end{exm}

\begin{exm}
We have $\langle(0),(2,0,0),(0),(0)\rangle=
\langle(0),(1,0),(0)\rangle \smallblacklozenge \langle(0),(0)\rangle$ and
$\langle(0),(1,1,0),(0),(0)\rangle=\langle(0),(0)\rangle\smallblacklozenge 
\langle(1,0),(0),(0)\rangle$. In Ex.~\ref{ex:G12a} and 
Figure~\ref{f:G12a} 
we considered the nested Catalan table
$\langle (2,0,0),(1,1,0),(0),(0),(0),(1,0),(0)\rangle
= \langle (2,0,0),(0), (0), (1,0),(0)\rangle 
\smallblacklozenge \langle (1,0),(0),(0)\rangle$.
Another example will be given in Ex.~\ref{exm:box}.
\end{exm}

\section{Chord diagrams with threads\label{sec:CD}}

For uncovering the combinatorial structure 
of (\ref{e:rr}) it was extremely helpful for us to have a 
graphical presentation as diagrams of chords and threads. 
To every term 
of the expansion (\ref{expansion}) of an $N$-point function we 
associate a diagram as follows: 
\begin{dfnt}[diagrammatic presentation]
Draw $N$ nodes on a circle, label them
from $b_0$ to $b_{N-1}$. Draw a (grey solid) chord 
between $b_r,b_s$ for every
factor $G_{b_rb_s}$ in (\ref{expansion}) and a (short-dashed
for $t,u$ even, long-dashed for $t,u$ odd) 
thread between 
$b_t,b_u$ for every factor $\frac{1}{E_{b_t}-E_{b_u}}$. 
The convention $t<u$ is chosen so that the diagrams come with a sign.
\end{dfnt}
\noindent 
It was already known in \cite{Grosse:2012uv} that the chords do not cross each
other (using cyclic invariance (\ref{e:rr_rot})) and that the
threads do not cross the chords (using  (\ref{rfp})). But the
combinatorial structure was not understood in \cite{Grosse:2012uv} and no
algorithm for a canonical set of chord diagrams could be given. The present
paper repairs this omission.

The $N/2=k+1$ chords in such a diagram divide the circle into $k+2$
pockets. The pocket which contain the arc segment between the
designated nodes $b_{0}$ and $b_{N-1}$ is by definition the root 
pocket $P_{0}$. Moving in the
counter-clockwise direction, every time a new pocket is entered it is
given the next number as index, as in Definition \ref{dfnt:CTG}. The
tree of these $k+2$ pockets, connecting vertices if the pockets border
each other, is the pocket tree. A pocket is called even (resp.\ odd)
if its index is even (resp.\ odd).

Inside  every even pocket, the short-dashed threads (between even nodes) 
form the direct tree, the long-dashed threads (between odd nodes) 
form the opposite tree. 
Inside every odd pocket, the short-dashed threads (between even nodes) 
form the opposite tree, the long-dashed threads (between odd nodes) 
form the direct tree. 

The sign $\tau$ of the diagram is given by 
\begin{equation}
\tau(T)=(-1)^{\sum_{j=1}^{k+1}e_{0}^{(j)}}\;,\label{e:sign}
\end{equation}
where $e^{(j)}_{0}$ is the first entry of the Catalan tuple 
corresponding to a pocket $P_{j}$. Indeed, 
for every pocket that is not a leaf or the root pocket, the chain of 
odd nodes starts with the highest index, which implies that every 
thread emanating from this node contributes a factor $(-1)$ to the
monomial (\ref{expansion}) compared with the lexicographic order
chosen there. In words: count for all pockets
other than the root pocket the total number $K$ of threads which go from
the smallest node into the pocket. The sign is even (resp.\ odd) if $K$ is
even (resp.\ odd). 

Figure \ref{f:G4} and \ref{f:G6} show nested Catalan tables and chord
diagrams of the $4$-point function and $6$-point function,
respectively.  Figure \ref{f:G12a} shows the chord diagram discussed
in Example~\ref{ex:G12a}.
\begin{figure}[!hpt]
\begin{picture}(120,30)\setlength{\unitlength}{1mm}
\put(25,3){\includegraphics[width=3cm]{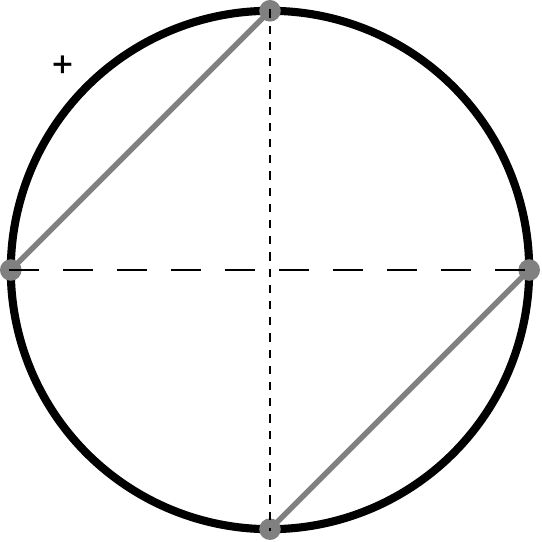}}
\put(65,3){\includegraphics[width=3cm]{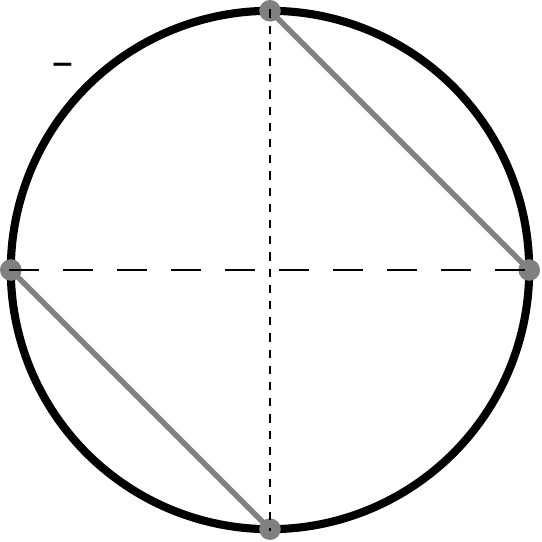}}
\put(30,0){\mbox{\scriptsize$\langle (1,0),(0),(0)\rangle$}}
\put(70,0){\mbox{\scriptsize$\langle (0),(1,0),(0)\rangle$}}
\end{picture}
\caption{\noindent The two chord diagrams and nested Catalan tables 
of $G^{(0)}_{b_{0}b_{1}b_{2}b_{3}}$.\label{f:G4}}
\end{figure}
\begin{figure}[!hpt]
\begin{picture}(120,108)\setlength{\unitlength}{1mm}
\put(35,77){\includegraphics[width=3cm]{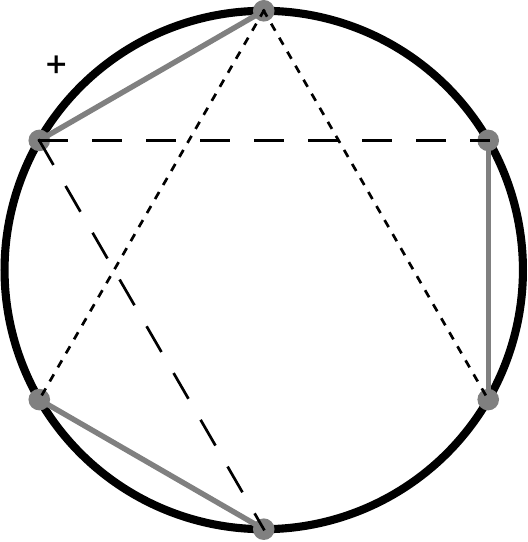}}
\put(70,77){\includegraphics[width=3cm]{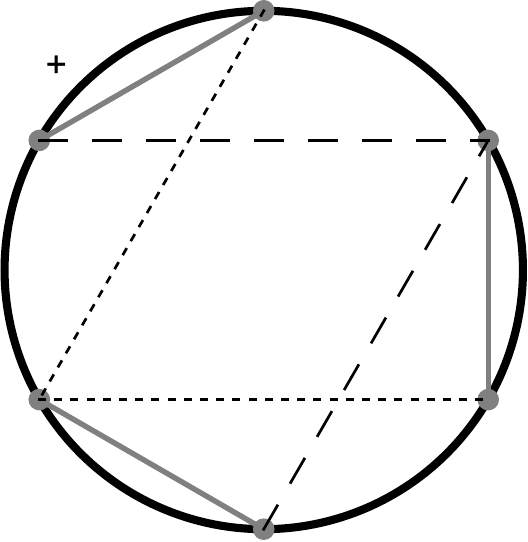}}
\put(35,74){\mbox{\scriptsize$\langle (2,0,0),(0),(0),(0)\rangle$}}
\put(70,74){\mbox{\scriptsize$\langle (1,1,0),(0),(0),(0)\rangle$}}
\put(35,40){\includegraphics[width=3cm]{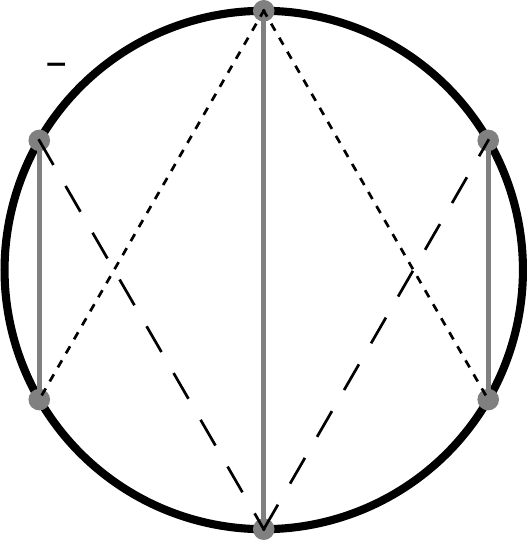}}
\put(70,40){\includegraphics[width=3cm]{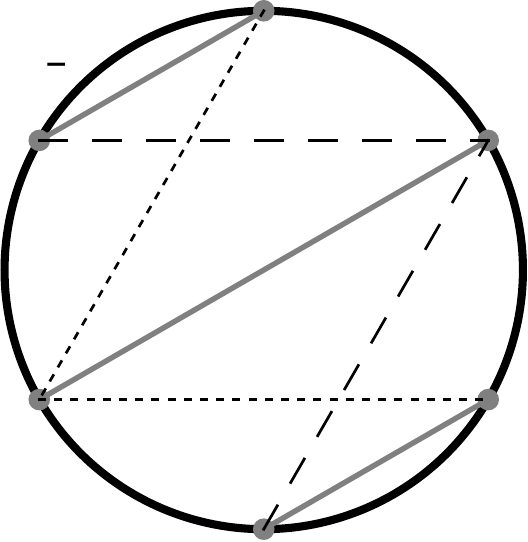}}
\put(35,37){\mbox{\scriptsize$\langle (1,0),(1,0),(0),(0)\rangle$}}
\put(70,37){\mbox{\scriptsize$\langle (1,0),(0),(1,0),(0)\rangle$}}
\put(17.5,3){\includegraphics[width=3cm]{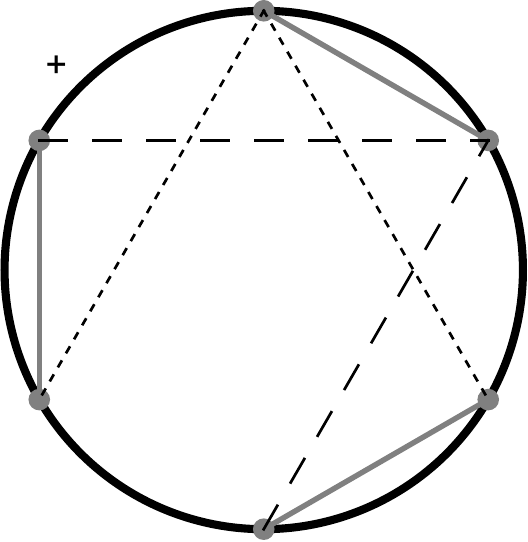}}
\put(52.5,3){\includegraphics[width=3cm]{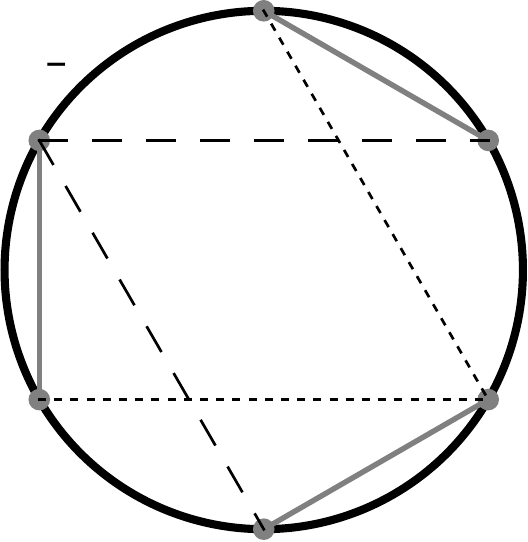}}
\put(87.5,3){\includegraphics[width=3cm]{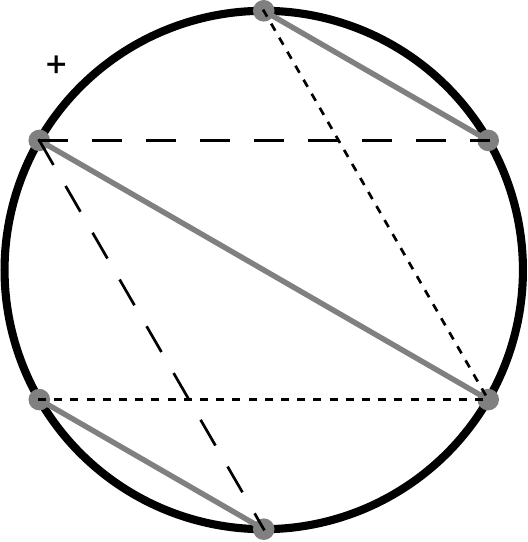}}
\put(17.5,0){\mbox{\scriptsize$\langle (0),(2,0,0),(0),(0)\rangle$}}
\put(52.5,0){\mbox{\scriptsize$\langle (0),(1,1,0),(0),(0)\rangle$}}
\put(87.5,0){\mbox{\scriptsize$\langle (0),(1,0),(1,0),(0)\rangle$}}
\end{picture}
\caption{The seven chord diagrams and nested Catalan tables of 
$G^{(0)}_{b_{0}b_{1}b_{2}b_{3}b_{4}b_{5}}$.\label{f:G6}}
\end{figure}
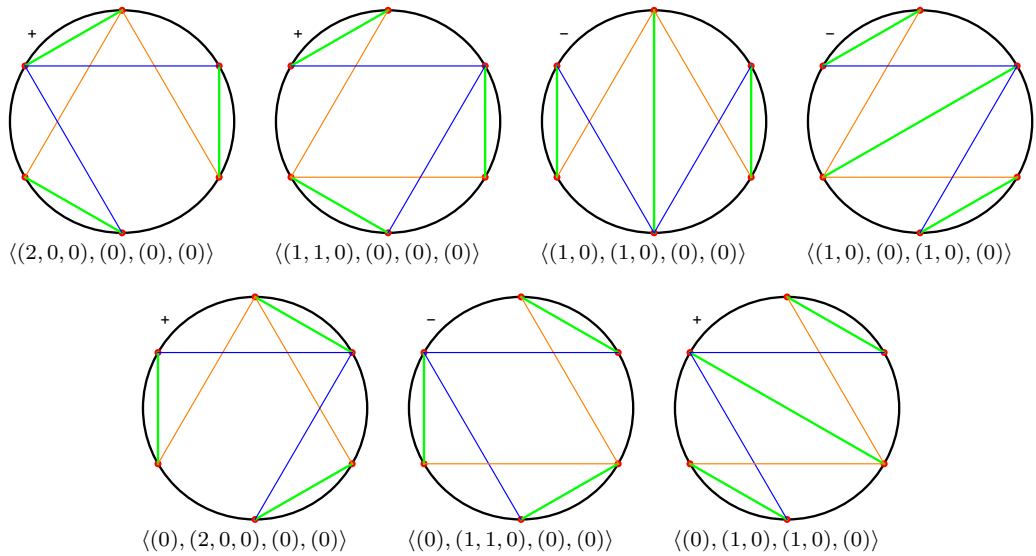
\begin{figure}[!hpt]
\begin{picture}(120,50)
\put(0,0){\includegraphics[width=5cm]{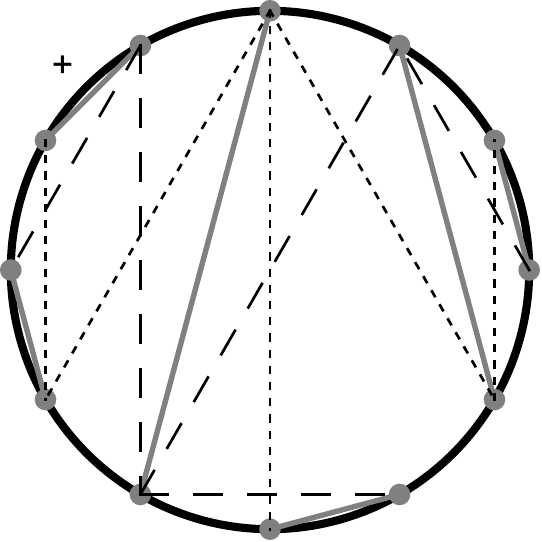}}
\put(60,20){\mbox{\footnotesize
$\langle(2,0,0),(1,1,0),(0),(0),(0),(1,0),(0)\rangle$}}
\end{picture}
\caption{A chord diagram and nested Catalan table contributing to 
  a planar $12$-point function $G^{(0)}_{b_0...b_{11}}$. Pocket tree and
  all non-trivial direct and opposite trees
have been given in Example~\ref{ex:G12a}.\label{f:G12a}}
\end{figure}

Now that a visual way to study the recursion relation (\ref{e:rr}) has
been introduced, it is much easier to demonstrate the concepts
introduced in Secs.~\ref{sec:CT} and \ref{sec:Ctab}.
\begin{exm}\label{exm:triangle}
  The operation $\smalllozenge\,$ is best demonstrated by an
  example:
\begin{equation*}
\langle(1,0),(0),(0)\rangle \smalllozenge 
\langle(0),(1,0),(0)\rangle=\langle(2,0,0),(0),(0),(1,0),(0)\rangle\;.
\end{equation*}
The corresponding chord diagrams are 
\[
\begin{picture}(120,32)
\put(0,3){\includegraphics[width=3cm]{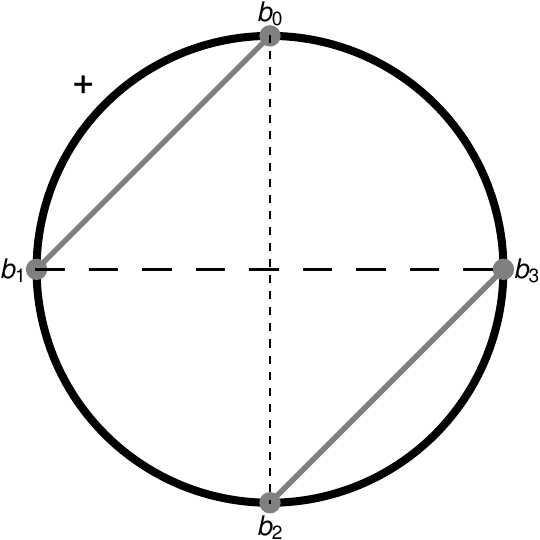}}
\put(34,16){\mbox{$\lozenge$}}
\put(40,3){\includegraphics[width=3cm]{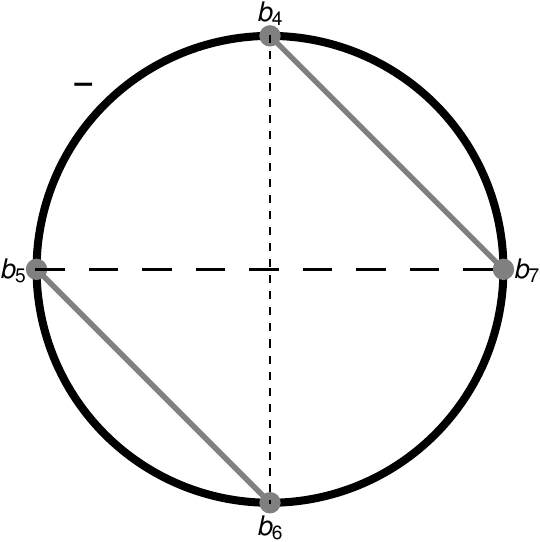}}
\put(78,16){\mbox{$=$}}
\put(87,3){\includegraphics[width=3cm]{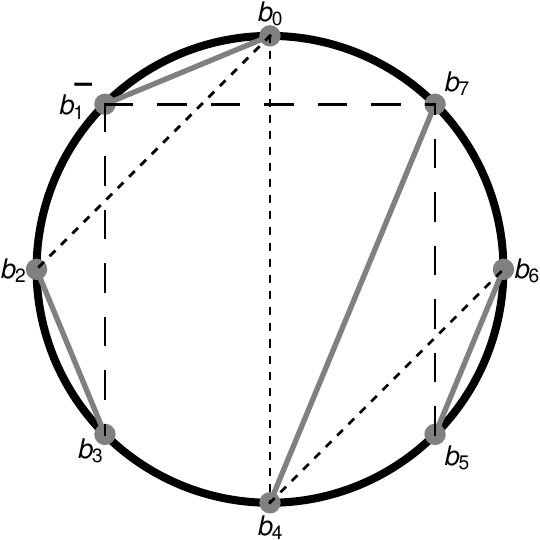}}
\put(3,0){\mbox{\scriptsize$\langle (1,0),(0),(0)\rangle$}}
\put(43,0){\mbox{\scriptsize$\langle (0),(1,0),(0)\rangle$}}
\put(85,0){\mbox{\scriptsize$\langle (2,0,0),(0),(0),(1,0),(0)\rangle$}}
\end{picture}
\]
The diagrammatic recipe is to cut both diagrams on the right side of
the designated node and paste the second into the first, where the
counter-clockwise order of the nodes must be preserved. Then both
designated nodes (here $b_0,b_4$) are connected by a short-dashed 
thread and nodes $b_{1}$
and $b_7=b_{N-1}$ by a long-dashed thread.

To $\smalllozenge$-decompose the nested Catalan table $\langle
(2,0,0),(0),(0),(1,0),(0)\rangle$, we first $\circ$-factorise 
the zeroth pocket $(2,0,0)$ via (\ref{circ-factor}).
Here $\sigma_{1}\big((2,0,0)\big)=1$ and, hence, $(2,0,0)=(1,0)\circ
(0)$. Next, we evaluate the number $\hat{k}$ defined in 
(\ref{k-lozenge}). We have 
$1+|\tilde{f}^{(0)}|=1$ and
$\sigma_{1}\big((3,0,0,1,0)\big)=2$. Consequently, 
we get from Definition~\ref{dfnt:triangle}
\begin{equation*}
\langle(2,0,0),(0),(0),(1,0),(0)\rangle=\langle (1,0),(0),(0)\rangle \smalllozenge \langle(0),(1,0),(0)\rangle\;.
\end{equation*}
\end{exm}

\begin{exm}\label{exm:box}
We employ the same example (with diagrams switched) to demonstrate 
the operation $\smallblacklozenge\,$. 
In terms of nested Catalan tables this becomes
\begin{equation*}
\langle(0),(1,0),(0)\rangle \smallblacklozenge 
\langle(1,0),(0),(0)\rangle=\langle(0),(2,1,0,0),(0),(0),(0)\rangle\;,
\end{equation*}
for which the chord diagrams are
\[
\begin{picture}(120,32)
\put(0,3){\includegraphics[width=3cm]{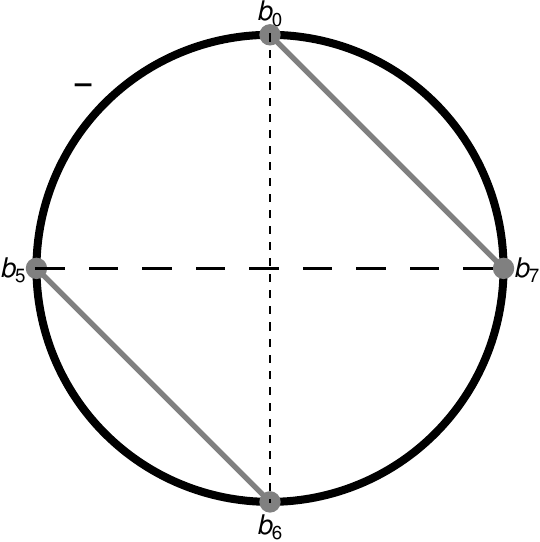}}
\put(34,16){\mbox{$\blacklozenge$}}
\put(40,3){\includegraphics[width=3cm]{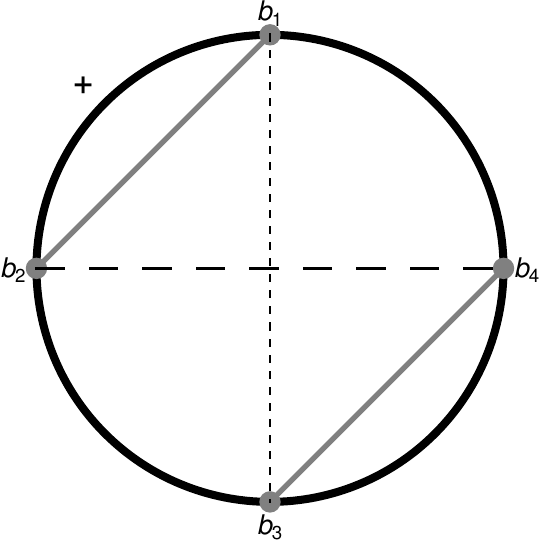}}
\put(78,16){\mbox{$=$}}
\put(87,3){\includegraphics[width=3cm]{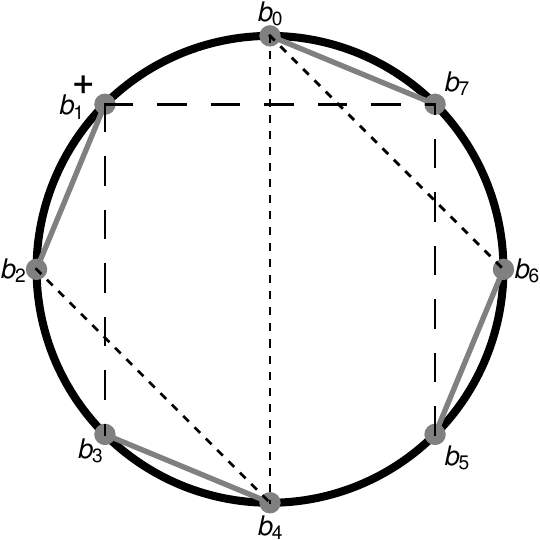}}
\put(3,0){\mbox{\scriptsize$\langle (0),(1,0),(0)\rangle$}}
\put(43,0){\mbox{\scriptsize$\langle (1,0),(0),(0)\rangle$}}
\put(85,0){\mbox{\scriptsize$\langle (0),(2,1,0,0),(0),(0),(0)\rangle$}}
\end{picture}
\]
The diagrammatic recipe is to cut the first diagram on the
left side of the designated node and the second diagram on the right
side. Then paste the second into the first,
where the counter-clockwise order of the nodes must be preserved. The
threads in the second diagram switch long/short doing so. Then, the
designated node of the first diagram is connected to the last node of
the second by a short-dashed thread, the designated node of the second
diagram is connected to the last node of the first diagram by a long-dashed
thread. 

Conversely, to $\smallblacklozenge$-decompose the nested Catalan table $\langle
(0),(2,1,0,0),(0),(0),(0)\rangle$, 
we first $\bullet$-factorise 
the first pocket $e^{(1)}=(2,1,0,0)$ via (\ref{bullet-factor}).
We have $e_{0}^{(1)}-1=1$, hence consider 
$\sigma_{1}\big((2,1,0,0)\big)=2$ and conclude 
$(2,1,0,0)=(1,0)\bullet (1,0)$. 
Next, we evaluate the number $\hat{l}$ in
(\ref{l-blacklozenge}). With 
$|\tilde{e}^{(0)}|+|\tilde{e}^{(1)}|+1=0+1+1=2$
the decomposition follows from $\sigma_{2}\big((1,3,0,0,0)\big)=2$ 
and yields 
\begin{equation*}
\langle (0),(2,1,0,0),(0),(0),(0)\rangle 
=\langle (0),(1,0),(0)\rangle \smallblacklozenge 
\langle(1,0),(0),(0)\rangle \;.
\end{equation*}
\end{exm}

\end{appendix}

\begin{ack}
We are grateful to an anonymous referee for an
exceptionally comprehensive report which contained numerous suggestions that 
improved this paper.
\end{ack}

\begin{funding}
This work was supported by the Deutsche Forschungsgemeinschaft via 
SFB 878 and the Cluster of Excellence\footnote{``Gef\"ordert
  durch die Deutsche Forschungsgemeinschaft (DFG) im Rahmen der
  Exzellenz\-strategie des Bundes und der L\"ander EXC 2044--390685587,
  Mathematik M\"unster: Dynamik--Geometrie--Struktur"} ``Mathematics
M\"unster''.
\end{funding}


%

------

\end{document}